\documentclass{llncs}

\usepackage[utf8]{inputenc}
\usepackage[T1]{fontenc}
\usepackage[usenames,dvipsnames]{color}
\usepackage{amsmath,amsfonts,stmaryrd}
\usepackage[linesnumbered,ruled,noend]{algorithm2e}
\usepackage{tikz,pgffor}
\usepackage{todonotes}

\usetikzlibrary{arrows}
\usetikzlibrary{shapes}
\usetikzlibrary{calc}
\usetikzlibrary{automata}
\usepackage{wrapfig}
\usepackage{stackrel}
\usepackage{enumitem}
\usepackage{eurosym}
\usepackage{multirow}

\newcommand{\Nset}{\mathbb{N}}
\newcommand{\Nseto}{\Nset_0}
\newcommand{\Qset}{\mathbb{Q}}
\newcommand{\Qsetp}{\mathbb{Q}_{>0}}
\newcommand{\Qsetpo}{\mathbb{Q}_{\ge 0}}
\newcommand{\Rset}{\mathbb{R}}
\newcommand{\Rsetp}{\mathbb{R}_{>0}}
\newcommand{\Rsetpo}{\mathbb{R}_{\ge 0}}

\newcommand{\maple}{\textsc{Maple}}

\renewcommand{\vec}[1]{\mathbf{#1}}
\newcommand{\action}[1]{\langle\!\langle #1 \rangle\!\rangle}

\newcommand{\eps}{\varepsilon}
\newcommand{\dist}{\mathcal{D}}

\newcommand{\Soff}{S_{\mathrm{off}}}
\newcommand{\Son}{S_{\mathrm{on}}}
\newcommand{\Sset}{S_{\mathrm{set}}}

\newcommand{\minit}{m_{\mathrm{in}}}
\newcommand{\timeouts}{\mathbf{d}}
\newcommand{\lowerto}{\ell}
\newcommand{\upperto}{u}

\newcommand{\probm}{\mathrm{Pr}}
\newcommand{\mcost}{\text{\euro}} 
\newcommand{\Ex}{\mathbb{E}}

\newcommand{\MP}{\mathrm{MP}}
\newcommand{\Rew}{\mathrm{Cost}}
\newcommand{\Time}{\mathrm{Time}}

\newcommand{\tp}[1]{\langle #1 \rangle}
\newcommand{\Act}{\mathit{Act}}

\newcommand{\eqdef}{\ensuremath{\stackrel{\text{\tiny def}}{=}}}
\newcommand{\Qmin}{Q_{\min}}
\newcommand{\tmin}{t_{\min}}
\newcommand{\tmax}{t_{\max}}
\newcommand{\rmax}{c_{\max}} 
\newcommand{\semiProb}{\Pi}

\newcommand{\calC}{\mathcal{C}}
\newcommand{\calA}{\mathcal{A}}
\newcommand{\calI}{\mathcal{I}}
\newcommand{\calJ}{\mathcal{J}}
\newcommand{\calN}{\mathcal{N}}
\newcommand{\calM}{\mathcal{M}}
\newcommand{\calR}{\mathcal{R}}

\newcommand{\sleep}{d_s}
\newcommand{\wakeup}{d_w}

\DeclareMathOperator*{\argmin}{argmin}
\DeclareMathOperator*{\argmax}{argmax}

\tikzstyle{ran}=[rounded rectangle,thick,draw,minimum size=1.4em,inner sep=.3ex]
\tikzstyle{act}=[rectangle,thick,draw,minimum size=1.4em,inner sep=1ex]
\tikzstyle{tran}=[thick,draw,->,>=stealth]

\newcommand{\appref}[1]{Appendix~\ref{#1}}
\begin{document}
\title{Mean-Payoff Optimization in Continuous-Time Markov Chains 
with Parametric Alarms\thanks{
	The authors are partly supported by the Czech Science Foundation, grant No.~15-17564S,
	by the DFG through the Collaborative Research Center SFB 912 -- HAEC,
	the Excellence Initiative by the German Federal and State Governments 
	(cluster of excellence cfAED), and the DFG-projects BA-1679/11-1 and BA-1679/12-1.}
}
	
\author{
	Christel Baier\inst{1}
	\and
	Clemens Dubslaff\inst{1}
	\and
	\v{L}ubo\v{s} Koren\v{c}iak\inst{2}
	\and \\
	Anton\'{\i}n Ku\v{c}era\inst{2}
	\and
	Vojt\v{e}ch {\v{R}}eh\'ak\inst{2}
}

\institute{
	TU Dresden, Germany\\
	\email{\{christel.baier, clemens.dubslaff\}@tu-dresden.de}
	\and
	Masaryk University, Brno, Czech Republic\\
	\email{\{korenciak, kucera, rehak\}@fi.muni.cz}
}

\maketitle

\begin{abstract}
Continuous-time Markov chains with alarms (ACTMCs) allow for alarm events that can be non-exponentially distributed. 
Within \emph{parametric} ACTMCs, the parameters of alarm-event distributions 
are not given explicitly
and can be subject of parameter synthesis.
An algorithm solving the $\varepsilon$-optimal 
parameter synthesis problem for parametric ACTMCs with long-run average optimization 
objectives is presented.
Our approach is based on reduction of the problem to finding long-run average optimal 
strategies in semi-Markov decision processes~(semi-MDPs) and sufficient discretization 
of parameter (i.e., action) space.
Since the set of actions in the discretized semi-MDP can be very large, a 
straightforward approach based on explicit action-space construction fails 
to solve even simple instances of the problem.
The presented algorithm
uses an enhanced policy iteration on symbolic representations of the action space.
The soundness of the algorithm is established for parametric ACTMCs with 
alarm-event distributions satisfying four mild assumptions that are shown 
to hold for uniform, Dirac and Weibull distributions in particular,
but are satisfied for many other distributions as well.
An experimental implementation shows that the 
symbolic technique substantially improves the efficiency of the synthesis 
algorithm and allows to solve instances of realistic size.
\end{abstract}

\section{Introduction}
\label{sec-intro}

\emph{Mean-payoff} is widely accepted as an appropriate concept for measuring long-run 
average performance of systems with rewards or costs. In this paper, we study the problem of synthesizing
parameters for (possibly \emph{non-exponentially} distributed) events in a given stochastic 
system to achieve an \mbox{$\varepsilon$-optimal} mean-payoff. One simple example of such events 
are \emph{timeouts} widely used, e.g., to  prevent deadlocks or to ensure some sort 
of progress in distributed systems. In practice, timeout durations are usually determined in an 
ad-hoc manner, requiring a considerable amount of expertise and experimental effort. This 
naturally raises the question of automating this design step, i.e., is there an algorithm 
synthesizing \emph{optimal} timeouts? 

The underlying stochastic model this paper relies on is provided by \emph{continuous-time Markov chains with alarms (ACTMCs)}. Intuitively, ACTMCs extend continuous-time Markov chains by generally distributed \emph{alarm events}, where at most one alarm is active during a system execution and non-alarm events can disable the alarm. In \emph{parametric ACTMCs}, every alarm distribution depends on one single parameter ranging over a given interval of eligible values.
For example, a timeout is a Dirac-distributed alarm event where the parameter specifies its duration.  A \emph{parameter function} assigning to every alarm a parameter value within the allowed interval yields a (non-parametric) ACTMC. 
 We aim towards an algorithm that synthesizes a parameter function for an arbitrarily small $\varepsilon > 0$ achieving \mbox{$\varepsilon$-optimal} mean-payoff.
\smallskip
~\\\textbf{Motivating example.}
To get some intuition about the described task, consider a dynamic 
power management of a disk drive inspired by \cite{QWP99}. The 
behavior of the disk drive can be described as follows (see 
Figure~\ref{fig-exa-drive}):
At every moment, the drive is either \emph{active} or \emph{asleep},
and it maintains a queue of incoming I/O operations of capacity~$N$.
The events of \emph{arriving} and \emph{completing} an I/O operation 
have exponential distributions with rates $1.39$ and $12.5$,
 respectively. When the 
queue is full, all newly arriving I/O operations are rejected. 
The I/O operations are performed only in the \emph{active} mode.
When the drive is \emph{active} and the queue becomes empty, an 
internal clock is set to $\sleep$. If then no further I/O request is 
received within the next $\sleep$ time units, 
the $\mathit{sleep}$ event changes the mode to \emph{asleep}. 
When the drive is \emph{asleep} and some I/O operation arrives, the 
internal clock is set to $\wakeup$
and after $\wakeup$ time the $\mathit{wakeup}$ event changes the mode to \emph{active}.
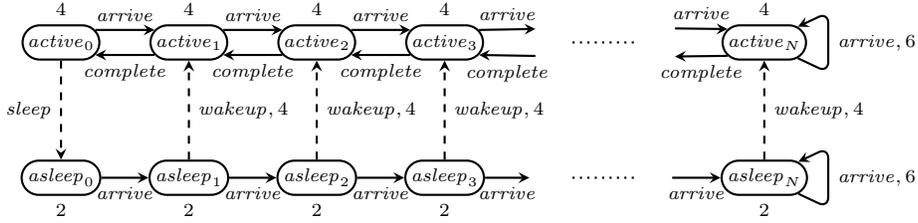
\begin{figure*}[t]\centering
	\vspace{-.5em}
	\clearpage{}\begin{tikzpicture}[x=1.7cm,y=2cm,font=\scriptsize]
	\foreach \x in {0,1,2,3}{%
	    \node[label={$4$}] (a\x) at (\x,0)   [ran] {$\mathit{active}_{\x}$};
	    \node[label=below:{$2$}] (s\x) at (\x,-0.9)  [ran] {$\mathit{asleep}_{\x}$};	
	}
	\node (a4) at (4,0)  [ran,draw=none] {\makebox[3.7em]{~}};	
	\node (s4) at (4,-0.9) [ran,draw=none] {\makebox[3.7em]{~}};	
	\node (a5) at (4.5,0)  [ran,draw=none] {\makebox[3.9em]{~}};	
	\node (s5) at (4.5,-0.9) [ran,draw=none] {\hspace*{3em}};	
	\node[label={$4$}] (a6) at (5.5,0)  [ran] {$\mathit{active}_N$};	
	\node[label=below:{$2$}] (s6) at (5.5,-0.9) [ran] {$\mathit{asleep}_N$};
	\foreach \x/\y in {0/1,1/2,2/3,3/4,5/6}{%
		\draw [tran,->] (a\x.20) -- node[above] {$\mathit{arrive}$} (a\y.160);
		\draw [tran,->] (s\x) -- node[below] {$\mathit{arrive}$} (s\y);
		\draw [tran,->] (a\y.200) -- node[below] {$\mathit{complete}$} (a\x.340);
    }
    \draw [tran,->,rounded corners] (a6) -- +(.5,-.2) --  node[right] {$\mathit{arrive},6$} +(.5,.2) -- (a6);
    \draw [tran,->,rounded corners] (s6) -- +(.5,-.2) --  node[right] {$\mathit{arrive},6$} +(.5,.2) -- (s6);
    \draw [thick,dotted] (a4.center) -- (a5.center);	
    \draw [thick,dotted] (s4.center) -- (s5.center);	
    \draw [tran,->,dashed] (a0) -- node[left] {$\mathit{sleep}$}(s0);	
    \foreach \x in {1,2,3,6}{%
        \draw [tran,->,dashed] (s\x) -- node[right] {$\mathit{wakeup},4$} (a\x);
	}
\end{tikzpicture}\clearpage{}
	\caption{Dynamic power manager of a disk drive.\label{fig-exa-drive}}
	\vspace{-.5em}
\end{figure*}
We annotate costs in terms of energy per time unit or instantaneous
energy costs for events.
The power consumption is $4$ and $2$ per time 
unit in the states \emph{active} and \emph{asleep}, respectively. 
Moving from \emph{asleep} to \emph{active} requires 
$4$ units of energy. Rejecting a newly arrived I/O request when the 
queue is full is undesirable, penalized by costs of $6$. 
All other transitions incur with cost $1$. 
Obviously, the designer of the disk drive controller has some 
freedom in choosing the delays $\sleep$ and $\wakeup$, i.e., they are free parameters of Dirac distribution. However, $\wakeup$ cannot 
be lower than the minimal time required to wake up the drive, which is 
constrained by the physical properties of the hardware used in the drive. Further, 
there is also a natural upper bound on $\sleep$ and $\wakeup$ given by the 
capacity of the internal clock. Observe that if $\sleep$ is too small, 
then many costly transitions from \emph{asleep} to \emph{active} are 
performed; and if $\sleep$ is too large, a lot of time is wasted in the 
more power consuming \emph{active} state. Similarly, if $\wakeup$ is too 
small, a switch to the \emph{active} mode is likely to be invoked 
with a few I/O operations in the queue, and more energy could have 
been saved by waiting somewhat longer; and if $\wakeup$ is too large, the 
risk of rejecting newly arriving I/O operations increases. Now we may 
ask the following instance of an optimal parameter synthesis problem we 
deal with in this paper:
\begin{quote}
	\emph{What values should a designer assign to the delays $\sleep$ and $\wakeup$ such that the long-run average power consumption is minimized?}
\end{quote}%
\textbf{Contribution.}
The main result of our paper is a \emph{symbolic} algorithm for \mbox{$\varepsilon$-optimal} 
parameter synthesis that is \emph{generic} in the sense 
that it is applicable to all systems where the optimized alarm events satisfy 
four abstractly formulated criteria.
We show that these criteria are fulfilled, e.g., for timeout events modeled by Dirac distributions, uniformly distributed alarms (used in, e.g., in variants of the CSMA/CD protocol \cite{Bertsekas1992}), and Weibull distributions (used to model hardware failures \cite{Weibull}). For a given $\varepsilon > 0$, our algorithm first computes a sufficiently small discretization step such that an \mbox{$\varepsilon$-optimal} parameter function exists even when its range is restricted to the discretized parameter values. Since the discretization step is typically very small, an \emph{explicit} construction of all discretized parameter values and their effects is computationally infeasible. Instead, our algorithm employs a symbolic variant of the standard policy iteration technique for optimizing the mean-payoff. It starts with \emph{some} parameter function  which is gradually improved until a fixed point is reached. In each improvement step, our algorithm computes a small \emph{candidate subset} of the discretized parameter values such that a possible improvement is realizable by one of these candidate values. This is achieved by designing a suitable \emph{ranking} function for each of the optimized events, such that an optimal parameter value 
is the minimal value of the ranking function in the interval of eligible parameter values. 
Then, the algorithm approximates the roots of the symbolic derivative of the ranking function, 
and constructs the candidate subset by collecting all discretized parameter values close to the 
approximated roots. This leads to a drastic efficiency improvement, which makes the resulting 
algorithm applicable to problems of realistic size.

\smallskip
~\\\textbf{Related work.} 
Synthesis of optimal timeouts guaranteeing quantitative properties in 
timed systems was considered in \cite{Diciolla2014}. 
There are various parametric formalisms for timed systems that deal 
with some sort of synthesis, such as  parametric timed automata~\cite{Jovanovic2014,Alur-TA-params,JLR:paramTA}, 
parametric one-counter automata~\cite{hkow-concur09}, parametric timed Petri nets~\cite{TLR:paramTimePetriNets}, or parametric Markov
models~\cite{DBLP:journals/sttt/HahnHZ11}.
However, all works referenced above do not 
consider models with continuous-time distributions, thus they synthesize different parameters than we do.
Contrary, the synthesis of appropriate 
rates in CTMCs was efficiently solved in \cite{Katoen-param-synt,CTMC-biochem,ceska,A:CTMCrateSynthesis}.
A special variant of ACTMC, where only alarms with Dirac distributions are allowed, has 
been considered in \cite{BKKNR:QEST2015,KRF:fdPRISM,KKR:MASCOTS_16}. Their algorithms
synthesize $\varepsilon$-optimal alarm parameters towards an expected reachability objective. 
Using a simulation-based approach, the optimization environment of the tool 
\textsc{TimeNET} is able to approximate locally
optimal distribution parameters in stochastic Petri nets, e.g., using methods as simulated annealing, hill climbing or 
genetic algorithms.
To the best of our knowledge, we present the first algorithm that approximates 
globally mean-payoff optimal parameters of non-exponential distributions in continuous-time models.

The (non-parametric) ACTMCs form a subclass of Markov regenerative processes 
(MRP) \cite{ABD:MarkovRegenProcess-qest14,CKT:MarkovRegenerativeSPN,regenerativeMC}.
Alternatively, ACTMCs can be also understood as a generalized 
semi-Markov processes (GSMPs) with at most one non-exponential event enabled in each 
state or as bounded stochastic Petri nets (SPNs)~\cite{Haas:book} %
with at most one non-exponential transition enabled in any reachable 
marking \cite{CKT:MarkovRegenerativeSPN}. 
Note that ACTMCs are analytically tractable thanks to methods for subordinated Markov-chain (SMC) that allow for efficient computation of transient and steady-state distributions \cite{DBLP:journals/pe/Lindemann93,CKT:MarkovRegenerativeSPN}.
Recently, methods for computing steady-state distributions 
in larger classes of regenerative GSMPs or SPNs have been presented in \cite{Vicario-MASCOTS16-steady-state}.
We did not incorporate this method into our approach as our 
methods to compute sufficiently small discretization and approximation precisions
to guarantee $\varepsilon$-optimal mean-payoffs are not directly applicable for this class of systems. 
To the best of our knowledge there are no efficient algorithms with a guaranteed error 
for computation of steady-state distribution for a general GSMP (or SPN).
For some cases it is even known that the steady-state distribution does not 
exist \cite{DBLP:conf/concur/BrazdilKKR11}.

\section{Preliminaries}
\label{sec-prelims}
Let $\Nset$, $\Nseto$, $\Qsetpo$, $\Qsetp$, 
$\Rsetpo$, and $\Rsetp$ denote the set of all positive integers, non-negative 
integers, non-negative rational numbers, positive rational numbers, non-negative 
real numbers, and positive real numbers, respectively. For a countable set $A$, 
we denote by $\dist(A)$ the set of discrete probability distributions over $A$, 
i.e., functions $\mu\colon A \to \Rsetpo$ where $\sum_{a\in A} \mu(a) = 1$.  
The \emph{support} of $\mu$ is the set of all $a \in A$ with $\mu(a) > 0$. 
A \emph{probability matrix} over some finite $A$ is a function 
$M\colon A{\times} A\rightarrow \Rsetpo$ where $M(a,\cdot)\in\dist(A)$ for all $a\in A$.
\subsection{Continuous-time Markov chains with alarms} 
\label{sec-CTMC}
A \emph{continuous-time Markov chain (CTMC)} is a triple $\calC=(S,\lambda,P)$, 
where $S$ is a finite set of states, $\lambda \in \Rsetp$ 
is a common exit rate\footnote{We can assume without restrictions that the parameter 
$\lambda$ is the same for all states of $S$ since every CTMC can be effectively transformed 
into an equivalent CTMC satisfying this property by the standard uniformization method 
(see, e.g., \cite{Norris:book}).}, and $P$ is a probability matrix over~$S$. Transitions 
in $\calC$ are exponentially distributed over the time, i.e., the probability of moving 
from $s$ to $s'$ within time $\tau$ is $P(s,s')\cdot(1-e^{-\lambda\cdot\tau})$. 

We extend CTMCs by generally distributed events called \emph{alarms}. 
A \emph{CTMC with alarms (ACTMC)} over a finite set of alarms $A$ is a tuple 
	\begin{center}
		$\calA\ \ =\ \ \big(S,\lambda,P,A,\tp{S_a},\tp{P_a},\tp{F_a}\big)$,
	\end{center}
where $(S,\lambda,P)$ is a CTMC and $\tp{S_a},\tp{P_a}$, and $\tp{F_a}$ are
tuples defined as follows: $\tp{S_a}=(S_a)_{a\in A}$ where $S_a$
is the set of states where an alarm $a\in A$ is enabled;
$\tp{P_a}=(P_a)_{a\in A}$ where $P_a$ is a probability matrix of some alarm $a\in A$
for which $P_a(s,s) = 1$ if $s \in S {\setminus} S_a$; 
and $\tp{F_a}=(F_a)_{a\in A}$ where $F_a$ 
is the cumulative distribution function (CDF) according to which the ringing time of an alarm $a \in A$ is distributed.
We assume that each distribution has finite mean
and $F_a(0) = 0$, i.e., a positive ringing time is chosen almost surely.
Furthermore, we require $S_a \cap S_{a'} = \emptyset$ for $a \neq a'$, i.e., 
in each state at most one alarm is enabled.
The set of states where some alarm is enabled is denoted by $\Son$, and we also 
use $\Soff$ to denote the set $S {\setminus} \Son$. 
The pairs $(s,s') \in S {\times}S$ with $P(s,s') > 0$ and $P_a(s,s') > 0$ 
are referred to as \emph{delay transitions} and 
\emph{$a$-alarm transitions}, respectively. 
\smallskip
~\\\textbf{Operational behavior.} Since in every state only one alarm is active,
the semantics of an ACTMC can be seen as an infinite CTMC amended with a timer
that runs backwards and is set whenever a new alarm is set or the alarm gets disabled.
A \emph{run} of the ACTMC $\calA$ is an infinite sequence 
$(s_0,\eta_0),t_0,(s_1,\eta_1),t_1,\ldots$ where $\eta_i$ is the current value of
the timer and $t_i$ is the time spent in $s_i$. If $s_0 \in \Soff$, then $\eta_0 = \infty$. 
Otherwise, $s_0 \in S_a$ for some $a \in A$ and the value of $\eta_0$ is selected randomly 
according to $F_a$. In a current configuration $(s_i,\eta_i)$, a random delay $t$ 
is chosen according to the exponential distribution with rate $\lambda$. 
Then, the time $t_i$ and the next configuration $(s_{i+1},\eta_{i+1})$ are determined as follows:
\begin{itemize}
	\item If $s_i \in S_a$ and $\eta_i \leq t$, then $t_i = \eta_i$ and $s_{i+1}$ is 
	selected randomly according to $P_a(s_i,\cdot)$. The value of $\eta_{i+1}$ 
	is either set to~$\infty$ or selected randomly according to $F_{b}$ for some
	$b\in A$, depending on whether the chosen $s_{i+1}$ belongs to $\Soff$ or 
	$S_{b}$, respectively (note that it may happen that $b = a$).  
	\item If $t < \eta_i$, then $t_i = t$ and $s_{i+1}$ is selected randomly according 
	to $P(s_i,\cdot)$. Clearly, if $s_{i+1} \in \Soff$, then $\eta_{i+1} =\infty$. 
	Further, if $s_{i+1} \in S_b$ and $s_i \not\in S_b$ for some $b \in A$, then 
	$\eta_{i+1}$ is selected randomly according to $F_{b}$. 
	Otherwise, $\eta_{i+1} = \eta_i - t$ (where $\infty - t = \infty$).  
\end{itemize}
Similarly as for standard CTMCs, we define a probability space over all 
runs initiated in a given $s_0 \in S$. We say that $\calA$ is \emph{strongly connected} if its underlying graph is, i.e., 
for all $s,s' \in S$, where $s \neq s'$, there 
is a finite sequence $s_0,\ldots,s_n$ of states such that $s = s_0$, 
$s' = s_n$, and $P(s_i,s_{i+1}) > 0$ or $P_a(s_i,s_{i+1}) > 0$ 
(for some $a \in A$) for all $0 \leq i < n$.

Note that the timer is set to a new value in a state $s$ only if $s \in S_a$ for some $a \in A$, and the previous state either does not belong to $S_a$ or the transition used to enter $s$ was an alarm transition\footnote{In fact, another possibility (which does not require any special attention) is that $s$ is the initial state of a run.}. Formally, we say that $s \in S_a$ is an \emph{$a$-setting state} if there exists $s' \in S$ such that either $P_b(s',s)>0$ for some $b \in A$, or $s' \not\in S_a$ and $P(s',s) > 0$. The set of all alarm-setting states is denoted by $\Sset$. If $\Sset\cap S_a$ is a singleton for each $a \in A$, we say that the alarms in $\calA$ are \emph{localized}. 
\smallskip
~\\\textbf{Cost structures and mean-payoff for ACTMCs.}
We use the standard cost structures that assign non-negative cost values to both states and transitions (see, e.g., \cite{Puterman:book}). More precisely, we consider the following cost functions: $\calR\colon S \rightarrow \Rsetpo$, which assigns a cost rate $\calR(s)$ to every state $s$ such that the cost $\calR(s)$ is paid for every time unit spent in $s$, and functions $\calI, \calI_{A}\colon S{\times}S \rightarrow \Rsetpo$ 
that assign to each delay transition and each alarm-setting transition, respectively, an instant execution cost. For
every run $\omega = (s_{0},\eta_{0}),t_0,(s_1,\eta_1),t_1,\dots$ of $\calN$ we define the associated \emph{mean-payoff} by
\[
\MP(\omega) \quad = \quad \limsup_{n \rightarrow \infty} \frac{\sum_{i=0}^n \big( \calR(s_i) \cdot t_i + \calJ(s_i,s_{i+1}) \big)}{\sum_{i=0}^n t_i} \,. 
\] 
Here, $\calJ(s_i,s_{i+1})$ is either $\calI(s_i,s_{i+1})$ or $\calI_A(s_i,s_{i+1})$ depending on whether $t_i < \eta_i$ or not, respectively. We use $\Ex[\MP]$ to denote the expectation of $\MP$. In general, $\MP$ may take more than one value with positive probability. However, if the graph of the underlying ACTMC is strongly connected, almost all runs yield the same mean-payoff value independent of the initial state~\cite{CKT:MarkovRegenerativeSPN}.

\subsection{Parametric ACTMCs} 
\label{ssec:pACTMC}
In ACTMCs, the distribution functions for
the alarms are already fixed. 
For example, if alarm $a$ is a timeout, it is 
set to some concrete value~$d$, i.e., the associated $F_a$ is a Dirac 
distribution such that $F_a(\tau) = 1$ for all $\tau \geq d$ and $F_a(\tau) = 0$ 
for all $0 \leq \tau < d$. 
Similarly, if $a$~is a random delay selected 
uniformly in the interval $[0.01,d]$, 
then $F_a(\tau) = 0$ for all $\tau < 0.01$ and \mbox{$F_a(\tau) = \min\{1,(\tau-0.01)/(d-0.01)\}$} 
for all $\tau \geq 0.01$. We also consider alarms with Weibull distributions, where $F_a(\tau) = 0$ for all $\tau \leq 0$ and \mbox{$F_a(\tau) = 1-e^{-(\tau/d)^k}$}
for all $\tau > 1$, where $k \in \Nset$ is a fixed constant. 

In the above examples, we can interprete $d$ as a \emph{parameter} and ask what parameter values minimize the expected long-run average costs.
For simplicity, we restrict our attention to distributions with only \emph{one} 
parameter.\footnote{In our current setting, distribution functions with several 
	parameters can be accommodated by choosing the parameter to optimize and fixing the others. In some cases we can also use simple extensions to synthesize, e.g., both $d_1$ and $d_2$ for the uniform distribution in  $[d_1,d_2]$ (see \appref{sec:dist}).} 
A \emph{parametric ACTMC} is defined similarly as an ACTMC, but instead of the concrete 
distribution function $F_a$, we specify a \emph{parameterized} distribution function 
$F_a[x]$ together with the interval $[\ell_a,u_a]$ of eligible parameter values for 
every $a \in A$. For every $d \in [\ell_a,u_a]$, we use $F_a[d]$ to denote the 
distribution obtained by instantiating the parameter $x$ with~$d$. 
Formally, a parametric ACTMC is a tuple  
\[
\calN\ \ =\ \ \big(S,\lambda,P,A,\tp{S_a},\tp{P_a},\tp{F_a[x]},\tp{\ell_a},\tp{u_a}\big)
\] 
where all components are defined in the same way as for ACTMC except 
for the tuples $\tp{F_a[x]}$, $\tp{\ell_a}$, and $\tp{u_a}$ of all $F_a[x]$, 
$\ell_a$, and $u_a$ discussed above. Strong connectedness, localized alarms, and cost structures
are defined as for (non-parametric) ACTMCs. 

A \emph{parameter function} for $\calN$ is a function $\timeouts\colon A \rightarrow \Rset$ such that 
$\timeouts(a) \in [\ell_a,u_a]$ for every $a \in A$. For every parameter function 
$\timeouts$, we use $\calN^{\timeouts}$ to denote the ACTMC obtained from $\calN$ 
by replacing each $F_a[x]$ with the distribution function $F_a[\timeouts(a)]$.
We allow only parametric ACTMCs that for each parametric function yield ACTMC.
When cost structures are defined on $\calN$,
we use $\Ex[\MP^{\timeouts}]$ to denote the expected mean-payoff in 
$\calN^{\timeouts}$. For a given $\varepsilon > 0$, we say that a parameter 
function $\timeouts$ is \emph{$\varepsilon$-optimal} if
\[
\Ex[\MP^{\timeouts}]\ \ \leq \ \ \inf_{\timeouts'} \ \Ex[\MP^{\timeouts'}] \ + \ \varepsilon,
\]
where $\timeouts'$ ranges over all parameter functions for $\calN$. 

\subsection{Semi-Markov decision processes}
\label{sec-semi-MDP}

A \emph{semi-Markov decision process (semi-MDP)} is a tuple 
$\calM = (M,\Act,Q,t,c)$, where $M$ is a finite set of states, 
$\Act = \biguplus_{m \in M} \Act_m$ is a set of actions 
where $\Act_m \neq \emptyset$ is a subset of actions enabled in 
a state $m$, $Q\colon \Act \rightarrow \dist(M)$ is a function assigning the probability $Q(b)(m')$ to move from $m\in M$ 
to $m'\in M$ executing $b \in \Act_m$, and functions $t,c\colon \Act \rightarrow \Rsetpo$ 
provide the expected time and costs when executing an action, respectively.\footnote{For our purposes, the actual distribution 
of the time and costs spent before executing some action is irrelevant, 
only their expectations matter, see Section~11.4 in~\cite{Puterman:book}.} 
A \emph{run} in $\calM$ is an infinite sequence 
$\omega = m_0,b_0,m_1,b_1,\ldots$ where $b_i \in \Act_{m_i}$ for every 
$i \geq 0$. The mean-payoff of $\omega$ is
\[
\MP(\omega)\ \  =\ \  \limsup_{n \rightarrow \infty} \ 
\left(\sum\nolimits_{i=0}^n c(b_i)\right)/\left(\sum\nolimits_{i=0}^n t(b_i)\right).
\]
A (stationary and deterministic) \emph{strategy} for $\calM$ is a 
function $\sigma\colon M \rightarrow \Act$  such that  $\sigma(m) \in \Act_m$ 
for all $m \in M$. \emph{Applying} $\sigma$ to $\calM$ yields the standard 
probability measure $\probm^\sigma$ over all runs initiated in a given initial 
state $\minit$. The expected mean-payoff achieved by $\sigma$ is denoted by 
$\Ex[\MP_{\calM}^{\sigma}]$. An \emph{optimal}\footnote{This strategy is optimal 
	not only among stationary and deterministic strategies, but even among all 
	randomized and history-dependent strategies.} strategy achieving the \emph{minimal} 
expected mean-payoff is guaranteed to exist, and it is computable by a simple 
\emph{policy iteration algorithm} (see, e.g., \cite{Puterman:book}).
\smallskip
~\\\textbf{\pmb{$\kappa$}-Approximations of semi-MDPs.}
Let $\calM = (M,\Act,Q,t,c)$ be a semi-MDP, and $\kappa\in \Qsetp$. We say that $Q^\kappa \colon \Act \rightarrow \dist(M)$ and \mbox{$t^\kappa,c^\kappa\colon \Act \rightarrow \Rsetpo$} are \emph{$\kappa$-approximations} of $Q$, $t$, $c$, if for all \mbox{$m,m' \in M$} and $b\in \Act_m$ it holds that $Q(b)$ and $Q^\kappa(b)$ have the same support, $\left|Q(b)(m')-Q^\kappa(b)(m')\right| \leq \kappa$,  $\left|t(b)-t^\kappa(b)\right|\leq\kappa$, and $\left|c(b)-c^\kappa(b)\right|\leq\kappa$. A \mbox{\emph{$\kappa$-approximation}} of 
$\calM$ is a semi-MDP $(M,\Act,Q^\kappa,t^\kappa,c^\kappa)$ where $Q^\kappa$, $t^\kappa$, $c^\kappa$ are $\kappa$-approximations of $Q$, $t$, $c$. We denote by $\left[\calM\right]_\kappa$ 
the set of all $\kappa$-approximations of $\calM$.

\section{Synthesizing $\varepsilon$-optimal parameter functions}
\label{sec-mdp}
In the following, we fix a strongly connected parametric ACTMC 
$\calN = (S,\lambda,P,A,\tp{S_a},\tp{P_a},\tp{F_a[x]},\tp{\ell_a},\tp{u_a})$ 
with localized alarms and cost functions $\calR$, $\calI$, and $\calI_A$, 
and aim towards an algorithm synthesizing an $\varepsilon$-optimal 
parameter function for $\calN$. 
Here, $\varepsilon$-optimality is understood with respect to the expected
mean-payoff. %
That is, we deal with the following computational problem:\smallskip

\noindent
\emph{$\varepsilon$-optimal parameter synthesis for parametric ACTMCs with localized alarms.}\\[1ex]
\makebox[4em][l]{\textit{Input:}} $\varepsilon\in\Qsetp$, a strongly connected parametric ACTMC $\calN$ with localized\\
\hspace*{4em} alarms, rational transition probabilities, rate $\lambda$, and cost functions $\calR$, \\ 
\hspace*{4em} $\calI$, and $\calI_A$.\\
\makebox[4em][l]{\textit{Output:}} An $\varepsilon$-optimal parameter function $\timeouts$.
\smallskip
\subsection{The set of semi-Markov decision processes $[\calM_\calN\langle\delta\rangle]_\kappa$}
\label{sec-reduction}
Our approach to solve the above problem is based on a reduction to 
the problem of synthesizing expected mean-payoff optimal strategies in semi-MDPs. 
Let $a\in A$, and let $s\in S_a \cap \Sset$. Recall that $\calN$ is localized
and thus, $s$ is the uniquely defined $a$-setting state.
Then, for every $d\in [\ell_a,u_a]$, consider runs initiated in a configuration $(s,\eta)$, 
where $\eta$ is chosen randomly according to $F_a[d]$. 
Almost all such runs eventually visit a \emph{regenerative} configuration $(s',\eta')$ where either
$s' \in \Soff$ or $\eta'$ is chosen randomly in $s'\in \Sset$, i.e., 
either all alarms are disabled or one is newly set.
We use $\semiProb_s(d)$ to denote the associated probability distribution over $\Sset\cup\Soff$, i.e.,
$\semiProb_s(d)(s')$ is the probability of visiting a regenerative configuration of the form $(s',\eta')$
from $s$ without previously visiting another regenerative configuration.
Further, we use $\mcost_s(d)$ and $\Theta_s(d)$ to denote the expected accumulated costs 
and the expected time elapsed until visiting a regenerative configuration, respectively. 
We use the same notation also for $s \in \Soff$, where $\semiProb_s(d) = P(s,\cdot)$, 
$\mcost_s(d) = \calR(s)/\lambda + P(s,\cdot)\cdot \calI_P$, and $\Theta_s(d) = 1/\lambda$ are independent of~$d$.
The semi-MDP $\calM_\calN=(\Sset\cup\Soff,\Act, Q, t, c)$ is defined over actions
\[
	\Act \ \ =\ \ \big\{\action{s,d}: d\in [\ell_a,u_a], s \in \Sset \cap S_a, a\in A\big\} 
		\cup \big\{\action{s,0}: s\in \Soff\big\},
\]
where for all $\action{s,d}\in\Act$  we have
$Q(\action{s,d}) = \semiProb_s(d)$, $t(\action{s,d}) = \Theta_s(d)$, and 
$c(\action{s,d}) = \mcost_s(d)$. Note that the action space of $\calM_\calN$ is
dense and that $\semiProb_s(d)$, $\Theta_s(d)$, and $\mcost_s(d)$ 
might be irrational. For our algorithms, we have to ensure a finite
action space and rational probability and expectation values.
We thus define the  $\delta$-discretization of $\calM_\calN$ as
$\calM_\calN\langle\delta\rangle=(\Sset\cup\Soff,\Act^\delta, Q^\delta, t^\delta, c^\delta)$
for a given discretization function $\delta\colon \Sset \rightarrow \Qsetp$. $\calM_\calN\langle\delta\rangle$
is defined as $\calM_\calN$ above, but over the action space $\Act^\delta=\bigcup_{s\in\Sset\cup\Soff}\Act^\delta_s$ with
\[
	\Act^\delta_s\ \ =\ \
		\big\{\action{s,d} :  d = \ell_a + i \cdot \delta(s) < u_a, i \in \Nseto \big\} \cup \big\{\action{s,u_a}\big\}
\]
for $s \in \Sset \cap S_a$ and $\Act^\delta_s=\left\{\action{s,0}\right\}$ otherwise.

To ensure rational values of $\semiProb_s(d)$, $\Theta_s(d)$, and $\mcost_s(d)$,
we consider the set of $\kappa$-approximations $\left[\calM_\calN\langle\delta\rangle\right]_\kappa$
of $\calM_\calN\langle\delta\rangle$ for any
$\kappa\in \Qsetp$. 
Note that, as $\calN$ is strongly connected, every 
$\calM \in \left[\calM_\calN\langle\delta\rangle\right]_\kappa$ is also 
strongly connected. 

\subsection{An explicit parameter synthesis algorithm}
\label{sec-algorithm}
Every strategy $\sigma$ minimizing the expected mean-payoff in $\calM_\calN$
yields an optimal parameter function $\timeouts^{\sigma}$ for $\calN$ 
defined by $\timeouts^{\sigma}(a) = d$ where $\sigma(s) = \action{s,d}$ 
for the unique $a$-setting state~$s$.
A naive approach towards an $\varepsilon$-optimal parameter function
minimizing the expected mean-payoff in $\calN$ is to compute a sufficiently
small discretization function $\delta$, approximation constant
$\kappa$, and some $\calM \in \left[\calM_\calN\langle\delta\rangle\right]_\kappa$
such that synthesizing an optimal strategy in $\calM$ yields an $\varepsilon$-optimal
parameter function for $\calN$. As $\calM$ is finite and contains only
rational probability and expectation values, the synthesis of an optimal strategy
for $\calM$ can then be carried out using standard algorithms for semi-MDP (see, e.g., \cite{Puterman:book}).
This approach is applicable under the following mild assumptions:
\begin{enumerate}
  \item \label{ass1} 
    For every $\varepsilon \in \Qsetp$, there are computable 
	$\delta\colon \Sset \rightarrow \Qsetp$ and $\kappa\in\Qsetp$ such that for every 
	$\calM \in \left[\calM_\calN\langle\delta\rangle\right]_\kappa$ and 
	every optimal strategy $\sigma$ for $\calM$, the associated parameter 
	function $\timeouts^{\sigma}$ is $\varepsilon$-optimal for $\calN$.
  \item \label{ass2}
  For all $\kappa\in\Qsetp$ and $s \in \Sset$, there are computable rational    %
  $\kappa$-approxi\-mations $\semiProb_s^\kappa(d)$, $\Theta_s^\kappa(d)$, 
  $\mcost_s^\kappa(d)$  of $\semiProb_s$, $\Theta_s$, $\mcost_s$. 	
\end{enumerate}
Assumption~\ref{ass1} usually follows from perturbation bounds on the expected mean-payoff using a straightforward error-propagation analysis. Assumption~\ref{ass2} can be obtained, e.g., by first computing 
$\kappa/2$-approximations of $\semiProb_s$, $\Theta_s$, and $\mcost_s$ for $s \in \Sset \cap S_a$, considering $a$ as alarm with 
Dirac distribution, and then integrate the obtained functions over the probability measure determined by
$F_a[x]$ to get the resulting $\kappa$-approximation (see also \cite{BKKNR:QEST2015,CKT:MarkovRegenerativeSPN}). 
Hence, Assumptions~\ref{ass1} and~\ref{ass2} rule out only those types of distributions that are rarely used in practice. In particular, the assumptions are satisfied for uniform, Dirac, and Weibull distributions. Note that Assumption~\ref{ass2} implies that for all $\delta\colon \Sset \rightarrow \Qsetp$ and $\kappa\in\Qsetp$, there is a computable $\calM \in \left[\calM_\calN\langle\delta\rangle\right]_\kappa$.  Usually, this naive \emph{explicit approach} to parameter synthesis is computationally infeasible due the large number of actions in~$\calM$.

\subsection{A symbolic parameter synthesis algorithm}
\label{sec-symbolic-algorithm}

Our symbolic parameter synthesis algorithm computes the set of states of some $\calM\in \left[\calM_\calN\langle\delta\rangle\right]_\kappa$ (see Assumption~\ref{ass1}) but avoids computing the set of all actions of~$\calM$ and their effects.
The algorithm is obtained by modifying the standard policy iteration \cite{Puterman:book} for semi-MDPs.

\paragraph{Standard policy iteration algorithm.}
When applied to $\calM$, standard policy iteration starts by picking an 
arbitrary strategy $\sigma$, which is then repeatedly 
improved until a fixed point is reached. In each iteration, the current 
strategy $\sigma$ is first evaluated by computing the associated \emph{gain} 
$g$  and  \emph{bias} $\textbf{h}$.\footnote{Here, it suffices to know that $g$ is a scalar and $\textbf{h}$ is a vector assigning numbers to states; for more details, see Sections~8.2.1 and 8.6.1 in~\cite{Puterman:book}.}
Then, for each state $s \in \Sset$, every outgoing action $\action{s,d}$ is ranked by the function
\begin{equation}
\label{eqn:Fnorm}\tag{${\times}$}
F_s^\kappa[g,\mathbf{h}](d)\ \ =\ \ \mcost^{\kappa}_s(d) - g \cdot \Theta^{\kappa}_s(d) + 
\semiProb^{\kappa}_s(d) \cdot  \mathbf{h}
\end{equation}
where $\mcost^{\kappa}_s$, $\Theta^{\kappa}_s$, and $\semiProb^{\kappa}_s$ are the determining functions of $\calM$. If the action chosen by $\sigma$ at $s$ does not have the best (minimal) rank, it is improved by redefining $\sigma(s)$ to some best-ranked action. The new strategy is then evaluated by computing its gain and bias and possibly improved again. The standard algorithm terminates when for all states the current strategy $\sigma$ is no improvement to the previous. 

\paragraph{Symbolic $\kappa$-approximations.} 
In many cases, $\semiProb_s(d)$, $\Theta_s(d)$, and $\mcost_s(d)$ 
for $s\in \Sset$ are expressible as infinite sums where the summands 
comprise elementary functions such as polynomials or $\exp(\cdot)$. 
Given $\kappa$, one may effectively truncate these infinite 
sums into finitely many initial summands such that the obtained expressions 
are differentiable in the interval $[\ell_a,u_a]$
and yield the \emph{analytical $\kappa$-approximations}
$\pmb{\semiProb}_s^\kappa(d)$, $\pmb{\Theta}_s^\kappa(d)$, and 
$\pmb{\mcost}_s^\kappa(d)$, respectively.
Now we can analytically approximate $F_s^\kappa[g,\mathbf{h}](d)$ by 
the value $\pmb{F}_s^\kappa[g,\mathbf{h}](d)$ obtained from (\ref{eqn:Fnorm})
by using the analytical 
$\kappa$-approximations: 
	\begin{equation}
	\label{eqn:Fapprox}\tag{$\star$}
	\pmb{F}_s^\kappa[g,\mathbf{h}](d) \ \ = \ \ \pmb{\mcost}^\kappa_s(d) - g \cdot \pmb{\Theta}^\kappa_s(d) + \pmb{\semiProb}^\kappa_s(d) \cdot  \mathbf{h}.
	\end{equation}
This function is differentiable for $d\in[\ell_a,u_a]$ when $g$ and $\mathbf{h}$ are constant.
Note that the discretized parameters minimizing $F_s^\kappa[g,\mathbf{h}](d)$ are either close to $\ell_a$, $u_a$, or roots of the derivative of $\pmb{F}_s^\kappa[g,\mathbf{h}](d)$.
Using the isolated roots and bounds $\ell_a$ and $u_a$, we identify a small set of candidate actions and explicitly evaluate only those instead of all actions.  
Note, that $\pmb{\semiProb}_s^\kappa(d)$, $\pmb{\Theta}_s^\kappa(d)$, $\pmb{\mcost}_s^\kappa(d)$ may return 
\emph{irrational} values for rational arguments. 
Hence, they cannot be evaluated precisely 
even for the discretized parameter values. However, when Assumption~\ref{ass2} is fulfilled,
it is safe to use \emph{rational} $\kappa$-approximations $\semiProb_s^\kappa(d)$, $\Theta_s^\kappa(d)$, 
$\mcost_s^\kappa(d)$ for this purpose. 
Before we provide our symbolic algorithm, we formally state the additional assumptions 
required to guarantee its soundness:

\begin{enumerate}
	\setcounter{enumi}{2}
	\item \label{ass3}
	For all $a \in A$, $s \in \Sset \cap S_a$, $\delta\colon \Sset \rightarrow \Qsetp$ and $\kappa\in\Qsetp$, 
	there are analytical \mbox{$\kappa$-approximations} $\pmb{\semiProb}_s^\kappa$, $\pmb{\Theta}_s^\kappa$, 
	$\pmb{\mcost}_s^\kappa$ of $\semiProb_s$, $\Theta_s$, $\mcost_s$, respectively, such that the function $\pmb{F}_s^\kappa[g,\mathbf{h}](d)$, where $g \in \Qset$ and $\mathbf{h}\colon \Sset \cup \Soff \rightarrow \Qset$ 
	are constant, is differentiable for $d\in[\ell_a,u_a]$. Further, 
	there is an algorithm approximating the roots of the derivative of 
	$\pmb{F}_s^\kappa[g,\mathbf{h}](d)$ in the interval $[\ell_a,u_a]$ up to the absolute error~$\delta(s)$.
	\item \label{ass4}
	For each $s \in \Sset$ (let $a$ be the alarm of $s$) there is a computable constant $\semiProb_s^{\min} \in \Qsetp$ 
	such that for all $d \in [\ell_a,u_a]$ and $s'\in \Sset \cup \Soff$
	we have that $\semiProb_s(d)(s') > 0$ implies $\semiProb_s(d)(s') \geq  \semiProb_s^{\min}$.
\end{enumerate}
Note that compared to Assumption~\ref{ass2}, the \mbox{$\kappa$-approximations} of Assumption~\ref{ass3} are 
harder to construct: we require closed forms for $\pmb{\semiProb}_s^\kappa$, $\pmb{\Theta}_s^\kappa$, and 
$\pmb{\mcost}_s^\kappa$ making the symbolic derivative of $\pmb{F}_s^\kappa[g,\mathbf{h}](d)$ 
computable and suitable for effective root approximation.

\begin{algorithm}[t]
	\SetAlgoLined
	\DontPrintSemicolon
	\SetKwInOut{Input}{input}\SetKwInOut{Output}{output}
	\SetKwData{n}{n}\SetKwData{f}{f}\SetKwData{g}{g}
	\SetKwData{Low}{l}\SetKwData{x}{x}
	\Input{A strongly connected parametric ACTMC $\calN$ with localized alarms, rational cost functions $\calR$, $\calI_P$, $\calI_{P_a}$, and $\varepsilon\in\Qsetp$ such that Assumptions~\ref{ass1}--\ref{ass4} are fulfilled.}
	\Output{An $\varepsilon$-optimal parameter function $\timeouts$.}
	\BlankLine
	compute the sets $\Sset$ and $\Soff$ \;
	compute $\delta$, $\kappa$, and $\semiProb_s^{\min}$ of Assumptions~\ref{ass1} and~\ref{ass4}\;
	let $\xi = \min\{\kappa/4,\semiProb_s^{\min}/3 : $ where $s \in \Sset\}$\label{alg:precision}\;
	fix the functions $\semiProb_s^\xi,\Theta_s^\xi,\mcost_s^\xi$ of Assumption~\ref{ass2} determining $\calM_{\xi} \in \left[\calM_\calN\langle\delta\rangle\right]_\xi$\;
	choose an arbitrary state $s' \in \Sset \cup \Soff$ and a strategy
	$\sigma'$ for $\calM_{\xi}$ \label{alg:line-choose-state} \;
	\Repeat{$\sigma$ = $\sigma'$} {
		$ \sigma := \sigma' $ \; 
		\BlankLine
		\tcp{policy evaluation}\BlankLine
		compute the \emph{gain}, i.e., the scalar $g := \Ex[\MP^{\sigma}]$ \;
		\BlankLine
		compute the \emph{bias}, i.e., the vector $\mathbf{h} \colon S \rightarrow \Qset$ satisfying
		$\mathbf{h}(s') = 0$ and for each $s\in \Sset \cup \Soff$,
		$\mathbf{h}(s) = \mcost^{\xi}_s(d) - g \cdot \Theta^{\xi}_s(d) + 
			\semiProb^{\xi}_s(d) \cdot  \mathbf{h}$, where $\sigma(s) = \action{s,d}$\;
		\BlankLine
		\ForEach{$a \in A$ and $ s \in \Sset \cap S_a$} {
			\tcp{policy improvement}\BlankLine
			compute the set $R$ of $\delta(s)/2$-approximations of the roots of the derivative of $\pmb{F}_s^{\xi}[g,\mathbf{h}](d)$ in~$[\ell_a,u_a]$\label{alg:roots} using %
			Assumption~\ref{ass3}\;
			\BlankLine
			$C := \left\{\sigma(s)\right\} \cup \left\{\action{s,d} \in \Act_s^\delta \colon |d - r| \leq 3\cdot\delta(s)/2, \text{ for }
			r \in R\cup\{\ell_a,u_a\}\right\}$\;\label{alg:cand}
			$\displaystyle B :=\argmin_{\action{s,d} \in C} ~  F_s^\xi[g,\mathbf{h}](d)$
			\label{alg:imprAction1}\;			
			\lIf{$\sigma(s)\in B$ }{ $\sigma'(s) := \sigma(s)$ \label{alg:imprAction2}}
			\lElse{$\sigma'(s):= \action{s,d}$ where $\action{s,d} \in B$}
			\label{alg:imprAction3}
		}
	}
	\Return $\timeouts^{\sigma}$
	\caption{Symbolic policy iteration}
	\label{alg:pol-iter}
\end{algorithm} 
\paragraph{Symbolic policy iteration algorithm.}
Algorithm~\ref{alg:pol-iter} closely mimics the standard policy iteration algorithm except for the definition of new precision $\xi$ at line~\ref{alg:precision} and the policy improvement part.
The local extrema points of 
$\pmb{F}_s^{\xi}[g,\mathbf{h}](d)$ (cf. Equation (\ref{eqn:Fapprox})) 
in the interval $[\ell_a,u_a]$ are 
identified by computing roots of its symbolic derivative (line~\ref{alg:roots}).
Then, we construct a small set $C$ of \emph{candidate actions} that are close to these roots and the bounds $\ell_a,u_a$ (line~\ref{alg:cand}). 
Each given candidate action is then evaluated using
the function $F_s^\xi[g,\mathbf{h}](d) = \mcost^{\xi}_s(d) - g \cdot \Theta^{\xi}_s(d) + 
\semiProb^{\xi}_s(d) \cdot  \mathbf{h}$
(cf. Equation (\ref{eqn:Fnorm})). An improving candidate action is chosen based on 
the computed values \mbox{(lines~\ref{alg:imprAction2}--\ref{alg:imprAction3})}. 

\begin{theorem}[Correctness of Algorithm~\ref{alg:pol-iter}]
	\label{thm:symbmainthm}
	The symbolic policy iteration algorithm effectively solves the 
	$\varepsilon$-optimal parameter synthesis problem for parametric ACTMCs and cost 
	functions that fulfill Assumptions~\ref{ass1}--\ref{ass4}.
\end{theorem}

\begin{proof}[Sketch]	Since the number of actions of $\calM_{\xi}$ is finite, Algorithm~\ref{alg:pol-iter}
	terminates.
	A challenging point is that we compute only approximate minima of the function $\pmb{F}_s^{\xi}[g,\mathbf{h}](d)$,
	which is \emph{different} from the function $F_s^{\xi}[g,\mathbf{h}](d)$ used to evaluate the candidate actions. There may exist an action that is not in the candidate set $C$ even if it has minimal $F_s^{\xi}[g,\mathbf{h}](d)$. Hence, the strategy 
	computed by Algorithm~\ref{alg:pol-iter} is not necessarily 
	optimal for~$\calM_\xi$. 
	Fortunately, due to Assumption~\ref{ass1}, the strategy induces $\varepsilon$-optimal parameters for any 
	parametric ACTMC if it is optimal for \emph{some} $\calM' \in \left[\calM_\calN\langle\delta\rangle\right]_\kappa$. 
	Therefore, for each $s \in \Sset$ we construct  $\semiProb_s'$, $\Theta_s'$, and $\mcost_s'$ determining such $\calM'$.
	Omitting the details,
	the functions $\semiProb_s'$, $\Theta_s'$, $\mcost_s'$ are constructed from 
	$\semiProb^{\xi}_s$, $\mcost^{\xi}_s$, $\Theta^{\xi}_s$ and
	slightly (by at most $2 \xi$) shifted 
	$\pmb{\semiProb}^{\xi}_s$, $\pmb{\mcost}^{\xi}_s$, $\pmb{\Theta}^{\xi}_s$.
	The constant $\xi$ was chosen sufficiently small such that
	the shifted $\pmb{\semiProb}^{\xi}_s$, $\pmb{\mcost}^{\xi}_s$, $\pmb{\Theta}^{\xi}_s$ are still $\kappa$-approximations of $\semiProb_s$, $\Theta_s$, $\mcost_s$
	and the shifted $\pmb{\semiProb}^{\xi}_s(d)(\cdot)$ is a correct distribution for each $d \in [\ell_a,u_a]$.
	The technical details of the construction are provided in \appref{sec:mainThmProof}. \qed
\end{proof}

The following theorem implies that the explicit and symbolic algorithms are applicable to parametric ACTMCs with uniform, Dirac, exponential, or Weibull distributions. The proof is technical, see \appref{sec:dist}. 

\begin{theorem}
	\label{lem:assumptions34}
	Assumptions~\ref{ass1}--\ref{ass4} are fulfilled for parametric ACTMCs with rational cost functions where for all 
	$a\in A$ we have that $F_a[x]$ is either a uniform, Dirac, exponential, or Weibull distribution.
\end{theorem} 

\section{Experimental evaluation}
\label{sec:experiments}
We demonstrate feasibility of the symbolic algorithm presented in Section~\ref{sec-mdp} 
on the running example of Figure~\ref{fig-exa-drive} and on a preventive maintenance model inspired by \cite{German-book}. 
The experiments were carried out\footnote{All the computations were run on a
machine equipped with Intel Core™ i7-3770 CPU processor
at 3.40~GHz and 8~GiB of DDR RAM.} using our prototype implementation of the symbolic algorithm 
implemented in $\maple$ \cite{maple}. 
$\maple$ is appropriate as it supports the root isolation of univariate polynomials with 
arbitrary high precision due to its symbolic engine. 
The implementation currently supports Dirac and uniform distributions only,
but could be easily extended by other distributions fulfilling Assumptions~\ref{ass1}--\ref{ass4}. 
\smallskip
~\\\textbf{Disk drive model.}
In the running example of this paper (see Section~\ref{sec-intro} and Figure~\ref{fig-exa-drive})
we aimed towards synthesizing delays $\sleep$ and $\wakeup$
such that the long-run average power consumption 
of the disk drive is $\varepsilon$-optimal. 
\begin{figure}[t]
	\centering
	\includegraphics[width=.48\textwidth]{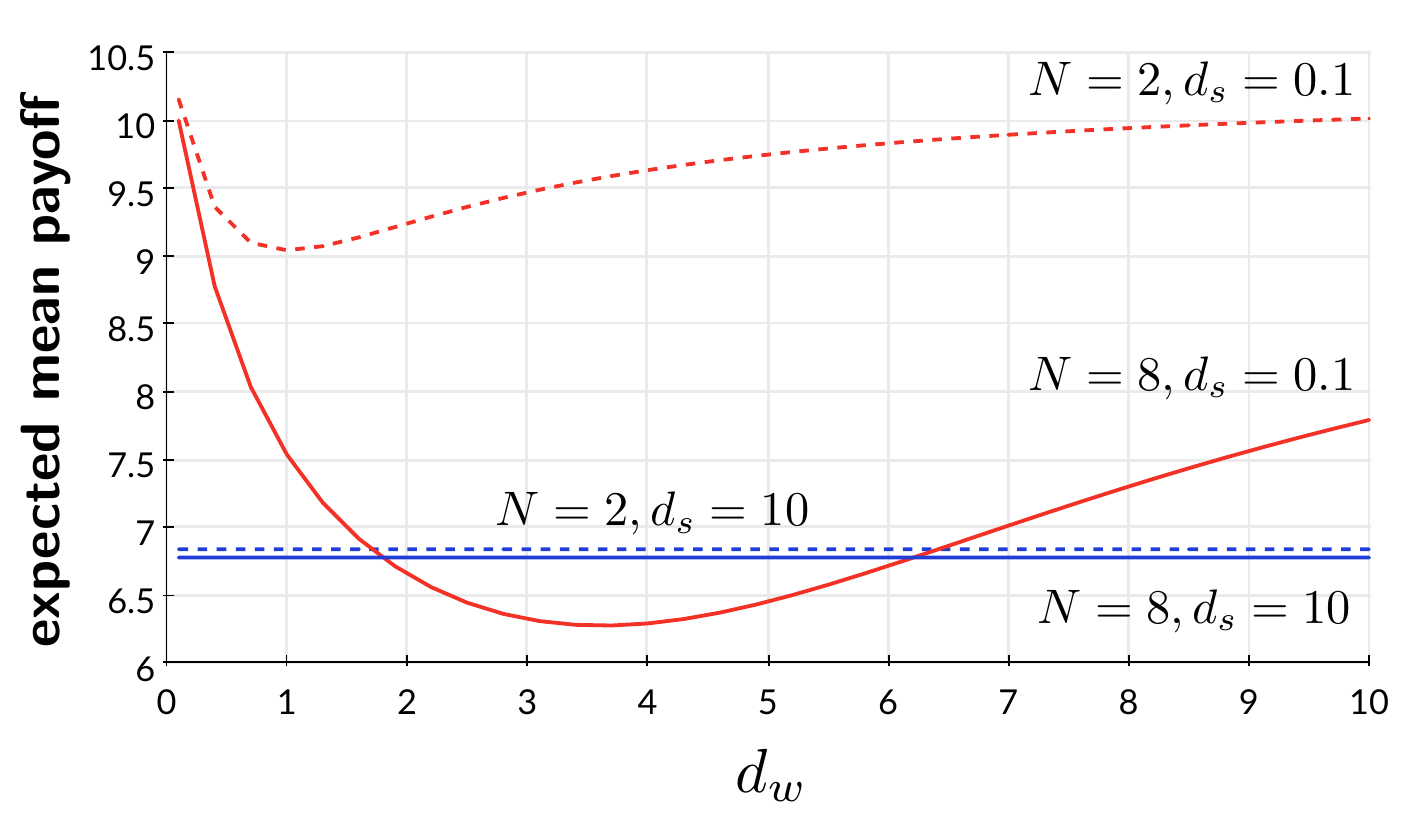}\hspace{1em}
	\includegraphics[width=.48\textwidth]{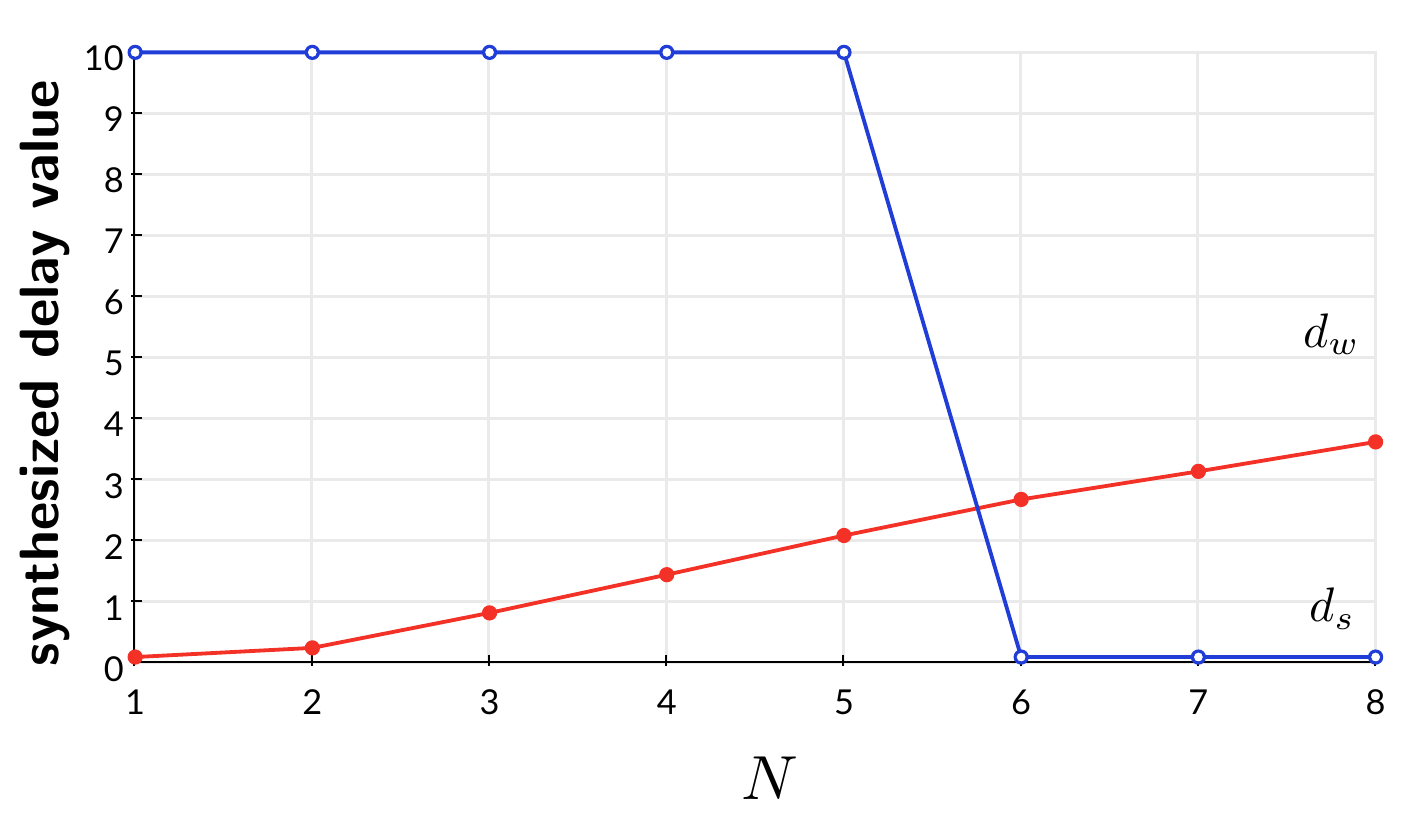}
	\caption{Results for the disk drive example: optimal expected mean-payoff (left),
		and trade-off illustrated by the synthesized delay values (right)}
	\label{fig-experiment}
\end{figure}
Let us describe the impact of choosing delay values $\sleep$ and
$\wakeup$ on the expected mean-payoff in more detail.
In Figure~\ref{fig-experiment}~(left), we illustrate the trade-off between
choosing different delays $\wakeup$ depending on delays
$\sleep\in\{0.1, 10\}$ and queue sizes $N\in\{2,8\}$.
When the queue is small, e.g., $N=2$ (dashed curves), 
the expected mean-payoff
is optimal for large $\sleep$ (here, $\sleep=10$). Differently,
when the queue size is large, e.g., $N=8$ (solid curves), 
it is better to choose
small $\sleep$ (here, $\sleep=0.1$) to minimize the expected 
mean-payoff with $\wakeup$ chosen at the minimum of the 
solid curve at around $3.6$.
This illustrates that the example is non-trivial.

The results of applying our synthesis algorithm for determining
$\varepsilon$-optimal delays $\sleep$ and $\wakeup$
depending on different queue sizes $N\in\{1,...,8\}$ 
with common delay bounds $\lowerto=0.1$ and $\upperto=10$ are depicted
in Figure~\ref{fig-experiment}~(right).
From this figure we observe that for increasing queue sizes, also
the synthesized value $\wakeup$ increases, whereas the optimal 
value for $\sleep$ is $\upperto$ in case $N<6$ and $\lowerto$ otherwise.

\begin{figure}%
	\scriptsize
  \hfill
  \begin{minipage}{.38\textwidth}
    \centering\clearpage{}\begin{tabular}{| c | l | r | r | r | }
\hline
 \multirow{2}{*}{$N$} & \multirow{2}{*}{$\eps$} & {creating} & {solving} & poly \\
 & & time [s] & time [s] & degree \\\hline
\multirow{4}{*}{2} & 
	0.1 		& 0.15 & 0.24 & 46\\
	& 0.01 		& 0.15 & 0.25 & 46\\
	& 0.001 	& 0.16 & 0.28 & 53\\
	& 0.0005 	& 0.16 & 0.33 & 53\\
\hline\multirow{4}{*}{4} & 
	0.1 		& 0.14 & 0.25 & 46\\
	& 0.01 		& 0.16 & 0.25 & 46\\
	& 0.001 	& 0.16 & 0.28 & 53\\
	& 0.0005 	& 0.16 & 0.33 & 53\\
\hline\multirow{4}{*}{6} & 
	0.1 		& 0.16 & 0.35 & 46\\
	& 0.01 		& 0.16 & 0.35 & 46\\
	& 0.001 	& 0.17 & 0.40 & 53\\
	& 0.0005 	& 0.18 & 0.43 & 53\\
\hline\multirow{4}{*}{8} & 
	0.1 		& 0.19 & 0.35 & 46\\
	& 0.01 		& 0.19 & 0.35 & 46\\
	& 0.001 	& 0.20 & 0.43 & 53\\
	& 0.0005 	& 0.22 & 0.44 & 53\\
\hline
\end{tabular}

\clearpage{}%
    \caption{Statistics of the symbolic algorithm applied to the disk drive
      example}%
    \label{tab:res-new}%
  \end{minipage}\hfill
  \begin{minipage}{.50\textwidth}
    \centering\clearpage{}\begin{tabular}{| c | l | r | r | r | r | }
\hline
 \multirow{2}{*}{$N$} & \multirow{2}{*}{$\eps$} & {creating} & {solving} & poly & \multirow{2}{*}{results}\\
 & & time [s] & time [s] & degree & \\\hline
\multirow{4}{*}{2} & 
	0.1 		& 0.15 & 1.80 & 86 & $\Ex[MP]$ 0.85524 \\
	& 0.01 		& 0.15 & 2.57 & 92 & $d_o$ 1.82752 \\
	& 0.001 	& 0.15 & 2.97 & 96 & $d_p$ 0.66167 \\
	& 0.0001 	& 0.15 & 3.84 & 101 & $d_q$ 2.05189 \\
\hline\multirow{4}{*}{4} & 
	0.1 		& 0.83 & 1.92 & 86 & $\Ex[MP]$ 0.46127 \\
	& 0.01 		& 0.92 & 2.40 & 92 & $d_o$ 1.92513 \\
	& 0.001 	& 1.04 & 3.06 & 97 & $d_p$ 0.66167 \\
	& 0.0001 	& 1.04 & 4.24 & 101 & $d_q$ 2.05189 \\
\hline\multirow{4}{*}{6} & 
	0.1 		& 2.25 & 2.18 & 87 & $\Ex[MP]$ 0.33060 \\
	& 0.01 		& 2.36 & 2.53 & 92 & $d_o$ 1.95764 \\
	& 0.001 	& 2.37 & 3.83 & 97 & $d_p$ 0.66167 \\
	& 0.0001 	& 2.41 & 4.41 & 101 & $d_q$ 2.05189\\
\hline\multirow{4}{*}{8} & 
	0.1 		& 17.08 & 2.22 & 87 & $\Ex[MP]$ 0.29536 \\
	& 0.01 		& 17.60 & 2.81 & 93 & $d_o$ 1.96540 \\
	& 0.001 	& 17.78 & 3.33 & 97 & $d_p$ 0.66167 \\
	& 0.0001 	& 17.87 & 4.48 & 102 & $d_q$ 2.05189 \\
\hline
\end{tabular}

\clearpage{}%
    \caption{Results and statistics of the symbolic algorithm applied to
    	 the preventive maintenance example}%
    \label{tab:res-rej}%
  \end{minipage}
  \hfill~
\end{figure}%

The table in Figure~\ref{tab:res-new} shows the running time of creation and solving of the $\maple$ models, as well as the largest polynomial degrees for selected 
queue sizes $N=\{2,4,6,8\}$ and error bounds 
$\varepsilon = \{0.1, 0.01, 0.001, 0.0005 \}$. 
In all cases, discretization step sizes of $10^{-6}\cdot 10^{-19}<\delta(\cdot)<10^{-19}$ 
were required to obtain results guaranteeing $\varepsilon$-optimal parameter functions. 
These small discretization constants underpin that the $\varepsilon$-optimal
parameter synthesis problem cannot be carried out using the explicit approach
(our implementation of the explicit algorithm runs out of memory for all of the listed instances). However, the
symbolic algorithm evaluating roots of polynomials with high
degree is capable to solve the problem within seconds in all cases. 
This can be explained through the small number of candidate actions we had to consider (always at most $200$).
\begin{figure*}[t]\centering
	\clearpage{}%
\begin{tikzpicture}[x=1.7cm,y=1.5cm,font=\scriptsize]

    \node[] (a2) at (2.0,0)   [ran] {$\mathit{normal}_{0}$};
    \node[] (a3) at (3,0)   [ran] {$\mathit{normal}_{1}$};
	\node[] (a6) at (5.5,0) [ran] {$\mathit{normal}_N$};	
    \draw [tran,->,rounded corners] (a6) -- +(.52,.2) --  node[right] {$2,6$} +(.52,-.2) -- (a6);
	\node (a4) at (4,0)    [ran,draw=none] {\makebox[3.7em]{~}};
	\node (a5) at (4.5,0)  [ran,draw=none] {\makebox[3.9em]{~}};	

    \node[] (s2) at (2,-1)  [ran] {$\mathit{degrad}_{0}$};
	\node[] (s3) at (3,-1) [ran] {$\mathit{degrad}_{1}$};
	\node[] (s6) at (5.5,-1) [ran] {$\mathit{degrad}_N$};
	\draw [tran,->,rounded corners] (s6) -- +(.52,.2) --  node[right,pos=.2] {$2,6$} +(.52,-.2) -- (s6);
	\node (s4) at (4,-1)   [ran,draw=none] {\makebox[3.7em]{~}};	
	\node (s5) at (4.5,-1) [ran,draw=none] {\hspace*{3.9em}};	

    \node[] (b2) at (2.0,-2)   [ran] {$\mathit{rejuven}$};
    \node[] (b3) at (3,-2)   [ran] {$\mathit{rej\_en}_{1}$};
    \node[] (b6) at (5.5,-2) [ran] {$\mathit{rej\_en}_N$};	
    \draw [tran,->,rounded corners] (b6) -- +(.52,.2) --  node[right] {$2,6$} +(.52,-.2) -- (b6);
    \node (b4) at (4,-2)    [ran,draw=none] {\makebox[3.7em]{~}};
    \node (b5) at (4.5,-2)  [ran,draw=none] {\makebox[3.9em]{~}};

	\node[] (t3) at (3,-3) [ran] {$\mathit{rej\_de}_{1}$};
	\node[] (t6) at (5.5,-3) [ran] {$\mathit{rej\_de}_N$};
	\draw [tran,->,rounded corners] (t6) -- +(.52,.2) --  node[right,pos=.8] 	{$2,6$} +(.52,-.2) -- (t6);
	\node (t4) at (4,-3)   [ran,draw=none] {\makebox[3.7em]{~}};	
	\node (t5) at (4.5,-3) [ran,draw=none] {\hspace*{3.9em}};

	\draw [tran,->] (a2.20) -- node[above] {$2$} (a3.160);
		\draw [tran,->] (s2.20) -- node[above] {$2$} (s3.160);
		\draw [tran,->] (a3.200) -- node[below] {$3$} (a2.340);
		\draw [tran,->] (s3.200) -- node[below] {$3$} (s2.340);

	\foreach \x in {a,b,s,t}{%
		\draw [tran,->] (\x3.20) -- node[above] {$2$} ++(0.5,0);
		\draw [tran,<-] (\x3.-20) -- node[below] {$3$} ++(0.5,0);
		\draw [tran,<-] (\x6.160) -- node[above] {$2$} ++(-0.5,0);
		\draw [tran,->] (\x6.200) -- node[below] {$3$} ++(-0.5,0);
    }

	\draw [tran,->] (b3) -- node[below] {$3$} (b2);
	\draw [tran,->] (t3) -- node[below,left] {$3$} (b2);

	\draw [tran,->,rounded corners,dashed] (s2) -- node[below right] {$o$} (b2);
	\path [->,tran]  (a2) edge[bend right=45, dashed] node[left,pos=0.22] {$o$} (b2.110);
	\path [->,tran]  (a3) edge[bend right=45, dashed] node[left,pos=0.22] {$o$} (b3.110);
	\path [->,tran]  (a6) edge[bend right=48, dashed] node[left,pos=0.22] {$o$} (b6.110);

	\path [->,tran]  (s3) edge[bend left=45, dashed] node[right,pos=0.78] {$o$} (t3);
	\path [->,tran]  (s6) edge[bend left=90, dashed] node[right,pos=0.78] {$o$} (t6);

	\node[] (u3) at (3,-4) [ran] {$\mathit{failed}$};

	\draw [tran,->,rounded corners] (s2)+(+0.1,-0.150) -- node[right,pos=0.51] {$1,0$} +(+0.1,-.5) -- +(4.5,-.5) -- +(4.5,-3) -- (u3);
	\draw [tran,->,rounded corners] (s3) -- node[left] {$1,4$} +(0,-.5) -- +(3.5,-.5) -- +(3.5,-3) -- (u3);
	\draw [tran,->,rounded corners] (s6) -- node[right] {$1,4N$} +(0,-.5) -- +(1,-.5) -- +(1,-3) -- (u3);
	\draw [tran,->,rounded corners] (t6) -- node[right] {$1,4N$} +(0,-1) --  (u3);
	\draw [tran,->,rounded corners] (t3) -- node[right] {$1,4$} (u3);

	\node[] (a0) at (0,0) [ran] {$\mathit{init\_ser}$};

	\draw [tran,->,rounded corners] (a0) -- +(-.2,+.35) --  node[above] 	{$2,1$} +(.2,.35) -- (a0);

	\node[] (b1) at (1,-2) [ran] {$\mathit{rejuven}$};	
	\node[] (s1) at (1,-1) [ran] {$\mathit{rej\_err}$};

	\draw [tran,->,rounded corners] (b2.192) -- node[below] {$2$} (b1.-12);
	\draw [tran,->,rounded corners] (b1) -- node[below] {$2$} +(-1,0) -- (a0);
	\draw [tran,->,rounded corners] (b2) -- node[right,pos=0.85] {$0.2$} (s1);
	\draw [tran,->,rounded corners] (b1) -- node[left] {$0.2$} (s1);

	\draw [tran,->,dashed] (b1.12) -- node[above] {$p$} (b2.168);

	\draw [tran,->, dashed]  (s1.295) -- node[below left] {$p$} (b2.153);

	\draw [tran,->,rounded corners] (s1) -- +(-.2,.35) --  node[above] 	{$2,1$} +(.2,.35) -- (s1);
	\draw [tran,->,rounded corners] (b1) -- +(-.2,-.3) --  node[below,pos=.9] 	{$2,1$} +(.2,-.3) -- (b1);
	\draw [tran,->,rounded corners] (b2) -- +(-.02,-.35) --  node[below] 	{$2,1$} +(.2,-.35) -- (b2);
	\draw [tran,->,rounded corners,dashed] (b2) -- +(-.08,-.35) --  node[below] 	{$p$} +(-.31,-.35) -- (b2);

	\node[] (u1) at (1,-4) [ran] {$\mathit{repair}$};	
	\node[] (u2) at (2,-4) [ran] {$\mathit{repair}$};	
	\node[] (t1) at (1,-3) [ran] {$\mathit{rep\_err}$};

	\draw [tran,->,rounded corners] (u1) -- node[below] {$1$} +(-1,0) -- (a0);

	\draw [tran,->,dashed] (u1.15) -- node[above] {$q$} (u2.165);

	\draw [tran,->, dashed]  (t1.295) -- node[below left] {$q$} (u2.155);

	\draw [tran,->,rounded corners] (u2.195) -- node[below] {$1$} (u1.-15);
	\draw [tran,->,rounded corners] (u2) -- node[right,pos=.8] {$0.1$} (t1);
	\draw [tran,->,rounded corners] (u1) -- node[left] {$0.1$} (t1);
	\draw [tran,->,rounded corners] (u3) -- node[below] {$3$} (u2);

	\draw [tran,->,rounded corners,dashed] (u2) -- +(-.2,.35) --  node[above] 	{$q$} +(.2,.35) -- (u2);
	\draw [tran,->,rounded corners] (t1) -- +(-.2,+.3) --  node[above,pos=.2] 	{$2,1$} +(.2,+.3) -- (t1);
	\draw [tran,->,rounded corners] (u1) -- +(-.2,-.35) --  node[below] 	{$2,1$} +(.2,-.35) -- (u1);
	\draw [tran,->,rounded corners] (u2) -- +(-.2,-.35) --  node[below] 	{$2,1$} +(.2,-.35) -- (u2);
	\draw [tran,->,rounded corners] (u3) -- +(-.2,-.35) --  node[below] 	{$2,1$} +(.2,-.35) -- (u3);

	\draw [tran,->,rounded corners] (a0) -- node[below]{$3$} (a2);

    \draw [thick,dotted] (a4.center) -- (a5.center);	
    \draw [thick,dotted] (b4.center) -- (b5.center);	
    \draw [thick,dotted] (s4.center) -- (s5.center);	
    \draw [thick,dotted] (t4.center) -- (t5.center);

    \draw [tran,->] (a2) -- node[right] {$1$} (s2);
    \foreach \x in {3,6}{%
        \draw [tran,->] (a\x) -- node[right] {$1$} (s\x);
        \draw [tran,->] (b\x) -- node[right] {$1$} (t\x);
	}
\end{tikzpicture}\clearpage{}
	\caption{Preventive maintenance of a server.}
	\label{fig:rejuvenation}
\end{figure*}
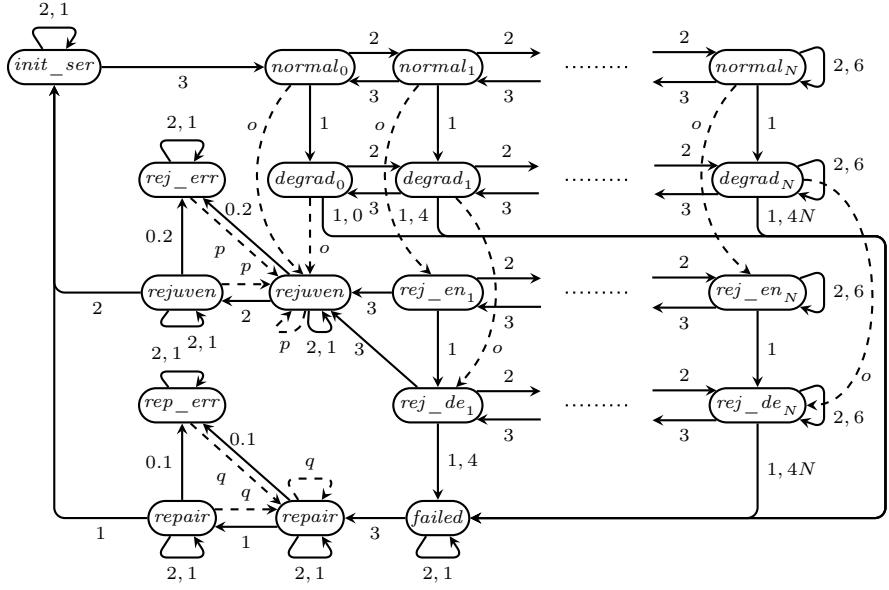
\smallskip
~\\\textbf{Preventive maintenance.} 
As depicted in Figure~\ref{fig:rejuvenation}, we consider a slightly modified model of a 
server that is susceptible to software faults \cite{German-book}.
A \emph{rejuvenation} is the process of performing \emph{preventive maintenance} of the server after a fixed period of time (usually during night time) to prevent performance degradation or even failure of the server.%
The first row of states in Figure~\ref{fig:rejuvenation} represents the normal behavior of the server.
Jobs arrive with rate $2$ and are completed with rate $3$.
If job arrives and queue is full, it is rejected what is penalized by cost $6$.
Degradation of server is modeled by delay transitions of rates $1$ leading to $\mathit{degrad}$ states 
of the second row or eventually leading to the $\mathit{failed}$ state.
The failure causes rejection of all jobs in the queue and incurs cost $4$ for each rejected job. 
After the failure is reported (delay event with rate $3$), the repair process 
is initiated and completed after two exponentially distributed steps of rate $1$.
The repair can also fail with a certain probability (rate $0.1$), thus after uniformly 
distributed time, %
the repair process is restarted. 
After each successful repair, the server is initialized by an exponential event with rate $3$.
The rejuvenation procedure is enabled after staying in $\mathit{normal}$ or $\mathit{degrad}$ states for time $d_o$.
Then the rejuvenation itself is initiated after all jobs in the queue are completed. 
The rejuvenation procedure behaves similarly as the repair process, except that 
it is two times faster (all rates are multiplied by two).

First, we want to synthesize the value of the delay after which the rejuvenation is enabled, 
i.e., we aim towards the optimal schedule for rejuvenation.
Furthermore, we synthesize the shifts $d_p$ and $d_q$ of the uniform distributions with length $2$ associated with rejuvenation and repair, respectively, 
i.e., the corresponding uniform distribution function is $F_x[d_x](\tau) = \min\{1,\max\{ 0, \tau-d_x/2 \} \}$, where $x \in \{p,q\}$.
The interval of eligible values is $[0.1,10]$ for all synthesized parameters. %
Similarly as for previous example we show results of experiments for queue sizes $N=\{2,4,6,8\}$ and error bounds 
$\varepsilon = \{0.1, 0.01, 0.001, 0.0001 \}$ in Table~\ref{tab:res-rej}.
The CPU time of model creation grows (almost quadratically) to the number of states, caused 
by multiplication of large matrices in $\maple$. 
As within the disk-drive example, we obtained the solutions very fast since we had 
to consider small number of candidate actions (always at most 500).
\smallskip
~\\\textbf{Optimizations in the implementation.}
For the sake of a clean presentation in this paper, we established \emph{global} theoretical upper bounds 
on $\delta$ and $\kappa$ sufficient to guarantee $\varepsilon$-optimal solutions, see \appref{sec:dist}.
The theoretical bounds assume the worst underlying transition structure of a given ACTMC.
In the prototype implementation, we applied some optimizations mainly computing \emph{local} 
upper bounds for each state in the constructed semi-MDP.
Also, to achieve better perturbation bounds on the expected mean-payoff, i.e., to compute bounds on expected time and cost to reach some state from all other states, we rely on techniques presented in \cite{BKKNR:QEST2015,KKR:MASCOTS_16}. 
Using these optimizations, for instance in the experiment of disk drive model, 
some discretization bounds $\delta$ %
were improved from $2.39 \cdot 10^{-239}$ to $7.03 \cdot 10^{-19}$. 
Note that even with these optimizations, the explicit algorithm for parameter
synthesis would not be feasible as, more than $10^{18}$ actions
would have to be considered for each state. This would clearly exceed the memory limit of 
state-of-the art computers.

\bibliographystyle{plain}
\bibliography{str-short.bib,paper.bib}

\newpage
\appendix

\section{Proof of Theorem~\ref{thm:symbmainthm}}
\label{sec:mainThmProof}

\medskip
\noindent
\textbf{Theorem~\ref{thm:symbmainthm} (Correctness of Algorithm~\ref{alg:pol-iter}).}
\textit{The symbolic policy iteration algorithm effectively solves the 
	$\varepsilon$-optimal parameter synthesis problem for parametric ACTMCs and cost 
	functions that fulfill Assumptions~\ref{ass1}--\ref{ass4}.
}
\medskip

\begin{proof}
	Since the number of actions in $\calM_{\xi}$ is finite, Algorithm~\ref{alg:pol-iter}
	terminates.
	A challenging point is that we compute only approximate minima of the function $\pmb{F}_s^{\xi}[g,\mathbf{h}](d)$,
	which might be \emph{different} from the function $F_s^{\xi}[g,\mathbf{h}](d)$ used to evaluate the candidate actions. There may exist an action that is not in the candidate set $C$ even if it has minimal $F_s^{\xi}[g,\mathbf{h}](d)$. Hence, the strategy 
	computed by Algorithm~\ref{alg:pol-iter} is not necessarily 
	optimal for~$\calM_\xi$. 
	Fortunately, due to Assumption~\ref{ass1}, the strategy induces $\varepsilon$-optimal parameters if it is optimal for \emph{some} $\calM' \in \left[\calM_\calN\langle\delta\rangle\right]_\kappa$. Therefore, it is sufficient to construct $\semiProb_s'$, $\Theta_s'$, and $\mcost_s'$ determining such a candidate $\calM'$ that is possibly different from $\calM_\xi$.
	
	Let us keep the notation as used in the main part of the paper. That is, we fix the parametric ACTMC $\calN$, MDP 
	$\calM_\xi \in \left[\calM_\calN\langle\delta\rangle\right]_\xi$, the strategy $\sigma$ returned by Algorithm~\ref{alg:pol-iter}, all functions $\semiProb^{\xi}_s$, $\mcost^{\xi}_s$, $\Theta^{\xi}_s$,
	$\pmb{\semiProb}^{\xi}_s$, $\pmb{\mcost}^{\xi}_s$, $\pmb{\Theta}^{\xi}_s$, and gain $g$ and bias $\vec{h}$ used during the last iteration of the outer cycle of Algorithm~\ref{alg:pol-iter}.

	For each state $s \in \Sset$ we define $\kappa$-approximations $\semiProb_s'$, $\Theta_s'$, $\mcost_s'$ (of $\semiProb_s$, $\Theta_s$, $\mcost_s$ on $\big\{ d :  \action{s,d} \in \Act^\delta_s \big\}$)  and $F_s'[g,\mathbf{h}](\cdot) \eqdef \mcost_s'(\cdot) - g \cdot \Theta_s'(\cdot) + \semiProb_s'(\cdot) \cdot  \mathbf{h}$ 
  such that 
	$$ F_s'[g,\mathbf{h}](d) \leq F_s'[g,\mathbf{h}](d'),$$
	for $d$ satisfying $\sigma(s) = \action{s,d}$ and every $d' \in \big\{ d'' :  \action{s,d''} \in \Act^\delta_s \big\}$.
	
	We will construct the functions $\semiProb_s'$, $\Theta_s'$, $\mcost_s'$ separately for each state $s \in \Sset$.
	Let $a$ be the unique alarm such that $s \in S_a$, $d$ be a rational number such that $\sigma(s) = \action{s,d}$, and $D \eqdef \big\{ d'' :  \action{s,d''} \in \Act^\delta_s \big\}$, i.e., $D = \big\{d'' :  d'' = \ell_a + i \cdot \delta(s) < u_a, i \in \Nseto \big\}\cup\{u_a\}$.
	Furthermore, let $C$ be the candidate set used to improve the strategy $\sigma$ for $s$ during the last iteration of the outer cycle in Algorithm~\ref{alg:pol-iter}.
	We define $C' \eqdef \{ d'' :  \action{s,d''} \in C \}$.
	
	First, we will construct $\semiProb_s'$, $\Theta_s'$, $\mcost_s'$ and  prove the theorem assuming that we have in hand certain ``shifted'' $\kappa$-approximations $\pmb{\semiProb}_s'$, $\pmb{\mcost}_s'$, $\pmb{\Theta}_s'$ that have some good properties.
	Then, we will construct the $\kappa$-approximations $\pmb{\semiProb}_s'$, $\pmb{\mcost}_s'$, $\pmb{\Theta}_s'$ and show that they have the needed properties.
	
	\vspace{0.4cm}
	Assume we can define ``shifted'' $\kappa$-approximations $\pmb{\semiProb}_s'$, $\pmb{\mcost}_s'$, $\pmb{\Theta}_s'$ (of $\semiProb_s$, $\mcost_s$, $\Theta_s$ on $[\ell_a,u_a]$) and $\pmb{F}_s'[g,\mathbf{h}](d') \eqdef \pmb{\mcost}_s'(d') - g \cdot \pmb{\Theta}_s'(d') + 
	\pmb{\semiProb}_s'(d') \cdot  \mathbf{h} $ for each $d' \in [\ell_a,u_a]$ such that 
	\begin{equation}
	\exists c \geq 0 : \forall d' \in [\ell_a,u_a] : \pmb{F}^{\xi}_s[g,\mathbf{h}](d') + c = \pmb{F}_s'[g,\mathbf{h}](d') \label{eq:shift}
	\end{equation}
	and 
	\begin{equation}
	\forall d' \in C': F_s^{\xi}[g,\mathbf{h}](d') \leq \pmb{F}_s'[g,\mathbf{h}](d'). \label{eq:overapprox}
	\end{equation} 
	From Assumption~\ref{ass3} and definition of $\pmb{F}_s'[g,\mathbf{h}]$ it follows that $\pmb{F}_s'[g,\mathbf{h}]$ is continuous and it has the same arguments of local extrema as $\pmb{F}^{\xi}_s[g,\mathbf{h}]$ on $[\ell_a,u_a]$. 
	From the definition of the candidate set $C$ (line~\ref{alg:cand} of Algorithm~\ref{alg:pol-iter}) and $C'$, it follows that $C' \cap \argmin_{d' \in D} \pmb{F}_s'[g,\mathbf{h}](d')$ is nonempty. 
	Let $\underline{d}$ be an arbitrary candidate of $C' \cap \argmin_{d' \in D} \pmb{F}_s'[g,\mathbf{h}](d')$.
	
	For each $d' \in D$ we set 
	\begin{itemize}
		\item $\semiProb_s'(d') \eqdef \semiProb_s^{\xi}(d')$, $\mcost_s'(d') \eqdef \mcost_s^{\xi}(d')$, $\Theta_s'(d') \eqdef \Theta_s^{\xi}(d')$ if $d' = d$,
		\item $\semiProb_s'(d') \eqdef \pmb{\semiProb}_s'(d')$, $\mcost_s'(d') \eqdef \pmb{\mcost}_s'(d')$, $\Theta_s'(d') \eqdef \pmb{\Theta}_s'(d')$ if $d' \neq d$, and
		\item $F_s'[g,\mathbf{h}](d') \eqdef \mcost_s'(d') - g \cdot \Theta_s'(d') + \semiProb_s'(d') \cdot  \mathbf{h}$.
	\end{itemize}
	Then for each $d' \in D$ we have
	$$ F_s'[g,\mathbf{h}](d) = F_s^{\xi}[g,\mathbf{h}](d) \leq \pmb{F}_s'[g,\mathbf{h}](\underline{d}) \leq \pmb{F}_s'[g,\mathbf{h}](d') = F_s'[g,\mathbf{h}](d').$$
	The left inequality follows from definition of $\pmb{F}_s'[g,\mathbf{h}]$ (see equation \eqref{eq:overapprox}), the right inequality follows from the definition of $\underline{d}$.

	\vspace{0.4cm}	
	To complete the proof, it remains to define $\kappa$-approximations $\pmb{\semiProb}_s'$, $\pmb{\mcost}_s'$, $\pmb{\Theta}_s'$ (of $\semiProb_s$, $\mcost_s$, $\Theta_s$ on $[\ell_a,u_a]$) such that they satisfy \eqref{eq:shift} and \eqref{eq:overapprox}.
	Let $\overline{d} \eqdef \argmax_{d'\in C'} F_s^{\xi}[g,\mathbf{h}](d') - \pmb{F}_s^{\xi}[g,\mathbf{h}](d')$.
	If $F_s^{\xi}[g,\mathbf{h}](\overline{d}) - \pmb{F}_s^{\xi}[g,\mathbf{h}](\overline{d}) \leq 0$
	we set $\Delta\semiProb \eqdef \vec{0}$, $\Delta\Theta \eqdef 0$, and $\Delta\mcost \eqdef 0$. Otherwise, we set
	\begin{itemize}
		\item $\Delta\semiProb \eqdef \semiProb_s^{\xi}(\overline{d}) - \pmb{\semiProb}_s^\xi(\overline{d})$,
		\item $\Delta\Theta \eqdef \Theta_s^{\xi}(\overline{d}) - \pmb{\Theta}_s^\xi(\overline{d})$, and 
		\item $\Delta\mcost \eqdef \mcost_s^{\xi}(\overline{d}) - \pmb{\mcost}_s^\xi(\overline{d})$.
	\end{itemize}
	Now, for all $d' \in D$, we put
	\begin{itemize}
		\item $\pmb{\semiProb}_s'(d') \eqdef \pmb{\semiProb}_s^\xi(d') + \Delta\semiProb$,
		\item $\pmb{\Theta}_s'(d') \eqdef \pmb{\Theta}_s^\xi(d') + \Delta\Theta$, and
		\item $\pmb{\mcost}_s'(d') \eqdef \pmb{\mcost}_s^\xi(d') + \Delta\mcost$.
	\end{itemize}
	
\noindent	We show that \eqref{eq:shift} and \eqref{eq:overapprox} hold:
	For each $d' \in [\ell_a,u_a]$ we have that 
	\begin{align*}
	\pmb{F}_s'[g,\mathbf{h}](d') &= \pmb{\mcost}_s'(d') - g \cdot \pmb{\Theta}_s'(d') + \pmb{\semiProb}_s'(d') \cdot  \mathbf{h} \\
	 &= (\pmb{\mcost}_s^\xi(d') + \Delta\mcost) - g \cdot (\pmb{\Theta}_s^\xi(d') + \Delta\Theta) + (\pmb{\semiProb}_s^\xi(d') + \Delta\semiProb) \cdot  \mathbf{h} \\
	 &= \pmb{F}_s^{\xi}[g,\mathbf{h}](d') + c,
	\end{align*}
	where $c = 0$ if $F_s^{\xi}[g,\mathbf{h}](\overline{d}) - \pmb{F}_s^{\xi}[g,\mathbf{h}](\overline{d}) \leq 0$ and $c=F_s^{\xi}[g,\mathbf{h}](\overline{d}) - \pmb{F}_s^{\xi}[g,\mathbf{h}](\overline{d}) > 0$ otherwise.
	Thus \eqref{eq:shift} holds. 
	Since $c = \max\{0, F_s^{\xi}[g,\mathbf{h}](\overline{d}) - \pmb{F}_s^{\xi}[g,\mathbf{h}](\overline{d})\}$ and $\overline{d} \eqdef \argmax_{d'\in C'} F_s^{\xi}[g,\mathbf{h}](d') - \pmb{F}_s^{\xi}[g,\mathbf{h}](d')$, for each $d' \in C'$ it holds that
	$$ F_s^{\xi}[g,\mathbf{h}](d') \leq  \pmb{F}_s^{\xi}[g,\mathbf{h}](d') + c = \pmb{F}_s'[g,\mathbf{h}](d'),$$
	what implies \eqref{eq:overapprox}.
	
	It remains to show that $\pmb{\semiProb}_s'$, $\pmb{\mcost}_s'$, $\pmb{\Theta}_s'$ are $\kappa$-approximations of $\semiProb_s$, $\mcost_s$, $\Theta_s$ on $[\ell_a,u_a]$: 
	Note, that for each $s' \in \Sset \cup \Soff$ it holds that $\Delta\semiProb(s') \leq 2\xi$, $\Delta\Theta \leq 2\xi$, and  $\Delta\mcost \leq 2\xi$. 
	Since $\xi \leq \kappa/4$, $\pmb{\Theta}_s'$ and $\pmb{\mcost}_s'$ are $\kappa$-approximations of $\Theta_s$ and $\mcost_s$ on $[\ell_a,u_a]$, respectively, and for each $d' \in [\ell_a,u_a]$ and $s' \in \Sset \cup \Soff$ it holds that $|\pmb{\semiProb}_s'(d')(s') - \semiProb_s(d')(s')| \leq 2\xi \leq \kappa$.
	
	Furthermore, for each $d' \in [\ell_a,u_a]$, $\pmb{\semiProb}_s'(d')$ is a distribution and it has the same support as $\semiProb_s(d')$:
	Let us fix $d' \in [\ell_a,u_a]$.
	Since $\pmb{\semiProb}^{\xi}_s$ and $\semiProb^{\xi}_s$ are $\xi$-approximations of $\semiProb_s$, then for each $s' \in \Sset \cup \Soff$ %
	it holds that
	if $\semiProb_s(d')(s')=0$ then $\pmb{\semiProb}^{\xi}_s(d')(s')=0$ and $\semiProb^{\xi}_s(d')(s') =0$ (if defined) and thus also $\pmb{\semiProb}_s'(d')(s')=0$.
	Moreover, since 
	$\xi \leq \semiProb_s^{\min}/3$ and 
	$|\pmb{\semiProb}_s'(d')(s') - \semiProb_s(d')(s')| \leq 2\xi$
	we have that if $\semiProb_s(d')(s') >0$ then $\pmb{\semiProb}_s'(d')(s') >0$ for each $s' \in \Sset \cap \Soff$.
	Thus, $\pmb{\semiProb}_s'(d')$ has the same support as $\semiProb_s(d')$.
	Finally, observe that $\sum_{s' \in \Sset \cup \Soff} \Delta\semiProb(s') =0$. 
	Thus, it holds that $\sum_{s' \in \Sset \cup \Soff} \pmb{\semiProb}_s'(d')(s') =1$.
	All the previous statements imply that $\pmb{\semiProb}_s'(d')$ is a distribution, thus $\pmb{\semiProb}_s'$ is $\kappa$-approximation of $\semiProb_s'$ on $[\ell_a,u_a]$.
	\qed
\end{proof} %
\section{Proof of Theorem~\ref{lem:assumptions34}}
\label{sec:dist}

Now we show that the Assumptions~\ref{ass1}--\ref{ass4} formulated in Section~\ref{sec-mdp} are satisfied for the $\varepsilon$-optimal parameter synthesis problem, where some concrete types of distributions $F_a[d]$ are used. The list is by no means exhaustive, and the studied distributions should be seen just as examples demonstrating the practical applicability of Algorithm~\ref{alg:pol-iter}.

\medskip
\noindent
\textbf{Theorem~\ref{lem:assumptions34}.}
\textit{	Assumptions~\ref{ass1}--\ref{ass4} are fulfilled for parametric ACTMC and cost functions if for all $a\in A$, $F_a[x]$ is
	either a uniform, Dirac, exponential, or Weibull distribution.
}
\medskip

First, we define new Assumptions \ref{app-assam-a} and \ref{app-assam-b} that formalize important discretization bounds on certain quantities of a parametric ACTMC. 
In Subsection~\ref{sec:assumption14}, we show that these two new assumptions 
imply Assumption~\ref{ass1} and \ref{ass4}. This we successfully use in the Subsections~\ref{sec:Dirac}, \ref{sec:uniform}, and \ref{sec:Weibull} where we separately show that Assumptions~\ref{ass2}, \ref{ass3}, \ref{app-assam-a}, and \ref{app-assam-b} are fulfilled for Dirac, uniform and Weibull distributions, respectively.

For the rest of this section, we fix a strongly connected parametric ACTMC 
$\calN = (S,\lambda,P,A,\tp{S_a},\tp{F_a[d]},\tp{\ell_a},\tp{u_a},\tp{P_a})$ 
and cost functions $\calR$, $\calI$, $\calI_A$, where all constants 
and functions are rational. Here, we formalize the following two \textbf{assumptions}:
\begin{enumerate}[label=\bf\Alph*)]
	\item \label{app-assam-a} For every $s \in \Sset$, there are effectively computable positive rational bounds 
	$$\semiProb_s^{\min}, \Theta_s^{\min}, 
	\Theta_s^{\max}, \text{ and }\mcost_s^{\max}$$ such that
	for all $d \in [\ell_a,u_a]$ where $a$ is the alarm of $s$, i.e., $s \in S_a$, we  have
	\begin{itemize}
		\item $\semiProb_s^{\min} \leq \semiProb_s(d)(s')$ for all $s' \in \Sset \cup \Soff$ where $\semiProb_s(d)(s') > 0$
		\item $\Theta_s^{\min} \leq \Theta_s(d) \leq \Theta_s^{\max}$
		\item $\mcost_s(d) \leq \mcost_s^{\max}$.
	\end{itemize} 
	\item \label{app-assam-b} For every $s \in \Sset$ and $\kappa \in \Qsetp$, 
	there is a discretization bound $\delta_{(s,\kappa)} \in \Qsetp$ such that for every $d,d' \in [\ell_a,u_a]$ and $|d - d'| \leq \delta_{(s,\kappa)}$ it holds that 
	\begin{itemize}
		\item $|\semiProb_s(d)(s') - \semiProb_s(d')(s')| \leq \kappa$ for all $s' \in \Sset \cup \Soff$, 
		\item $|\Theta_s(d) - \Theta_s(d')| \leq \kappa$, and
		\item $|\mcost_s(d) - \mcost_s(d')| \leq \kappa$  
	\end{itemize}
	where $a$ is the alarm of $s$, i.e., $s \in S_a$.
\end{enumerate}

Note that for every parametric ACTMC we have $\semiProb_s(d)(s')=0$ either for all parameters $d$ or for none of them.
Intuitively, Assumption~\ref{app-assam-b} connects $\kappa$-approximation and $\delta$-discretization.
We will show the connection more formally.

Assume we can find small enough $\kappa$ such that we cause at most $\varepsilon$ error in the expected mean-payoff when applying a $2\kappa$-approximation.
We will keep one $\kappa$-approximation and dedicate the other one for the errors caused by discretization.

Let $\delta\colon \Sset \to \Qsetp $ be a function obtained from Assumption~\ref{app-assam-b} by setting $\delta(s)\eqdef \delta_{(s,\kappa)}$ for each $s\in\Sset$.
We show that an optimal strategy $\sigma$ of $\calM \in \left[\calM_\calN\langle\delta\rangle\right]_{\kappa}$ induces an $\varepsilon$-optimal strategy $\timeouts^{\sigma}$ for $\calN$:
Let $\semiProb_s^{\kappa}$, $\Theta_s^{\kappa}$, $\mcost_s^{\kappa}$ be the determining functions of $\calM$. 
We define determining functions $\semiProb_s'$, $\Theta_s'$, $\mcost_s'$ of $\calM' \in \left[\calM_\calN\right]_{2\kappa}$ as follows:
for each $d \in [\ell_a,u_a]$ and $s' \in \Sset \cup \Soff$ let
$\semiProb_s'(d)(s') \eqdef \semiProb_s^{\kappa}(\delta'(d))(s')$, 
$\Theta_s'(d) \eqdef \Theta_s^{\kappa}(\delta'(d))$, and
$\mcost_s'(d) \eqdef \mcost_s^{\kappa}(\delta'(d))$, 
where $\delta'(d) = \ell_a + i \cdot \delta(s)$ for $d \in \big[\ell_a + i \cdot \delta(s), \ell_a + (i+1) \cdot \delta(s)\big)$ and $i\in \Nseto$.
Clearly, $\sigma$ is an optimal strategy of $\calM'$ and $\timeouts^{\sigma}$ is $\varepsilon$-optimal of $\calN$, what follows from our setting of $\kappa$.

Thus we concentrate on computation of sufficiently small $\kappa$ for each $\varepsilon > 0$, what is addressed in the next subsection.

\subsection{Computation of $\kappa$}
\label{sec:assumption14}

Since the mean-payoff is defined as a fraction, we first connect the error bound of a fraction with perturbation bounds of its numerator and denominator.
Intuitively, in order to guarantee an error $\varphi$ of a 
fraction, the numerator and the denominator have to be computed with a certain precision.

\begin{lemma}\label{lem:frac_bouds}
	For every $a,b,a',b',\overline{a},\overline{b},\underline{b},\varphi \in\Rsetp$ such that $a\leq\overline{a}$, and	$\underline{b} \leq b\leq\overline{b}$, we have that
	\[ \text{ if both }
	|a -a'| \text{ and } |b-b'| \text{ are } \leq \frac{\underline{b}^2\cdot
		\varphi}{\overline{a}+\overline{b}+\overline{b}\cdot\varphi} \text{ then }
	\left|\frac{a}{b} - \frac{a'}{b'}\right|  \leq  \varphi
	\]

\end{lemma}

\begin{proof}
	Let $\varphi'$ be $(\underline{b}^2\cdot
		\varphi ) / (\overline{a}+\overline{b}+\overline{b}\cdot\varphi)$.
	We divide the proof into the following sub-cases:
	\begin{itemize}
		\item[(i)] Let $a \geq a'$ and $b\leq b'$.  \\Then $a/b \geq a'/b'$.
		Due to $|a -a'| \leq \varphi'$, we have that $a'\geq a-\varphi'$
		and, due to $|b-b'| \leq \varphi'$, we have that $b'\leq
		b+\varphi'$. Also note that $$\varphi'= \frac{\underline{b}^2\cdot
			\varphi}{\overline{a}+\overline{b}+\overline{b}\cdot\varphi} \leq
		\frac{b^2\cdot \varphi}{a+b}.$$ Hence,
		\begin{align*}
		&\left| \frac{a}{b} -  \frac{a'}{b'}\right|   
		= \frac{a}{b} - \frac{a'}{b'}
		\leq \frac{a}{b} - \frac{a-\varphi'}{b+\varphi'}
		= \frac{ab+ a\varphi' - ba + b\varphi'}{b^2 + b\varphi'}
		= \frac{a\varphi' + b\varphi'}{b^2 + b\varphi'} =\\ 
		&= \frac{a + b}{b^2/\varphi' + b}
		\leq \frac{a + b}{b^2/\varphi' }
		\leq \frac{a + b}{b^2/((b^2\cdot \varphi) /(a+b) )}
		= \varphi
		\end{align*}
		\item[(ii)] Let $a \leq a'$ and $b \geq b'$. \\
		Then $a/b \leq a'/b'$.
		Due to $|a -a'| \leq \varphi'$, we have that $a'\leq a+\varphi'$
		and, due to $|b-b'| \leq \varphi'$, we have that $b'\geq
		b-\varphi'$. Also note that $$\varphi'= \frac{\underline{b}^2\cdot
			\varphi}{\overline{a}+\overline{b}+\overline{b}\cdot\varphi} \leq
		\frac{b^2\cdot \varphi}{a+b+b\cdot\varphi}$$ Hence,
		\begin{align*}
		\left| \frac{a}{b} -  \frac{a'}{b'}\right|  
		= \frac{a'}{b'} - \frac{a}{b}
		\leq \frac{a+\varphi'}{b-\varphi'} - \frac{a}{b}
		= \frac{ab+ \varphi'b - ba + \varphi'a}{b^2 - b\varphi'}
		= \frac{\varphi'b + \varphi'a}{b^2 - b\varphi'} = \\
		= \frac{b + a}{b^2/\varphi' - b} 
		\leq \frac{b + a}{b^2/((b^2\cdot \varphi) /(a+b+b\cdot\varphi) ) - b}
		= \frac{b + a}{ (a+b)/\varphi + b - b}
		= \varphi
		\end{align*}
		\item[(iii)] Let $a < a'$ and $b < b'$. \\ If $a/b \geq a'/b'$ then $ \big|
		\frac{a}{b} - \frac{a'}{b'}\big| = \frac{a}{b} - \frac{a'}{b'} \leq
		\frac{a'}{b} - \frac{a'}{b'} \leq \varphi$ due to (i).  \\If
		$a/b < a'/b'$ then $ \big| \frac{a}{b} - \frac{a'}{b'}\big| =
		\frac{a'}{b'} - \frac{a}{b} \leq \frac{a'}{b} - \frac{a}{b} \leq
		\varphi$ due to (ii).
		\item[(iv)] Let $a > a'$ and $b > b'$. \\If $a/b \geq a'/b'$ then $ \big|
		\frac{a}{b} - \frac{a'}{b'}\big| = \frac{a}{b} - \frac{a'}{b'} \leq
		\frac{a}{b'} - \frac{a'}{b'} \leq \varphi$ due to (i). \\ If
		$a/b < a'/b'$ then $ \big| \frac{a}{b} - \frac{a'}{b'}\big| =
		\frac{a'}{b'} - \frac{a}{b} \leq \frac{a}{b'} - \frac{a}{b} \leq
		\varphi$ due to (ii).
	\end{itemize}\qed
\end{proof}

The next lemma establishes a perturbation bound on $\kappa$
such that the expected mean-payoff achieved by a given strategy changes among $\kappa$-approximations at 
most by a given~$\varepsilon > 0$. Note that in this lemma, we use only the bounds specified in Assumption~\ref{app-assam-a} for $s\in\Sset$. Bounds for other states can be computed accordingly. Hence, 
Lemma~\ref{lem-semiMDP} finishes the proof that Assumption~\ref{ass1} holds if Assumptions~\ref{app-assam-a} and \ref{app-assam-b} are satisfied.

\begin{lemma}
	\label{lem-semiMDP}
	Let $\calM$ be a strongly connected semi-MDP, $\sigma$ be a strategy such that $\calM$ stays 
	strongly connected when the set of actions is restricted to those selected by $\sigma$. 
        Let $\calM'$ be a $\kappa$-approximation of $\calM$. For every error $\varepsilon>0$, it holds that
	$\left|\Ex[\MP_{\calM}^{\sigma}] - \Ex[\MP_{\calM'}^{\sigma}]\right| \leq \varepsilon$ if
	\[
	\kappa\ \leq \ \min \left\{ 
	\begin{array}{l}
	\frac{(\tmin/2)^2 \cdot \frac{\varepsilon}{n}}%
	{2 w_{\max} \cdot (2/\Qmin)^n \cdot (2+\frac{\varepsilon}{n}) \cdot (1+2n w_{\max} (2/\Qmin)^n)}, 
	\frac{\Qmin}{2}, \frac{\tmin}{2}, \frac{\rmax}{2} 
	\end{array} 
	\right\}
	\]
	where $\Qmin,\tmin,\tmax,\rmax \in \Qsetp$ are bounds on the minimal probability, 
	minimal and maximal time step, and maximal costs occurring
	in $\calM$, $n$ is the number of states of $\calM$, and  $w_{\max} = \max\{\rmax,\tmax\}$.
\end{lemma}
\begin{proof}
    Let $\calM = (M,\Act,Q,t,c)$ 
	be a semi-MDP with $M = \{m_1,\ldots,m_n\}$ and $\sigma$ some 
	strategy for $\calM$. Note that every 
	$\kappa$-approximation $\calM'$ of $\calM$ can be obtained via a sequence 
	of semi-MDPs $\calM_1,\ldots,\calM_{n+1}$ where $\calM_1 = \calM$, $\calM_{n+1} =\calM'$, 
	and every $\calM_{i+1}$ is obtained from $\calM_i$ by modifying only the values 
	of $Q(b)$, $t(b)$, and $c(b)$ with $b \in \Act_{m_i}$. 
	Note that, the bounds $\Qmin,\tmin,\tmax,\rmax \in \Qsetp$ may not hold for all $ \calM_i $ due to the changes in the previous semi-MDPs. 
	Hence, we require 
	$$\forall m \in M : \kappa \leq \min\{\Qmin,\tmin,\rmax\}/2.$$ 
	and then the new constants $\Qmin'=\Qmin/2$, $\tmin'=\tmin/2$, $\tmax'=2\tmax$, and $\rmax'=2\rmax$ correctly bound all semi-MDPs in the sequence. Formally, for each $ 1 \leq i \leq n+1$, 
	$\calM_i = (M,\Act,Q_i,t_i,c_i)$  satisfies
	\begin{itemize}
		\item $Q_i(b)(m') >0$ implies $Q_i(b)(m') \geq \Qmin/2 = \Qmin'$, 
		\item $\tmin'=\tmin/2 \leq t_i(b) \leq 2\tmax = \tmax'$, 
		\item $c_i(b) \leq 2\rmax = \rmax'$
	\end{itemize}
	for all $m,m' \in M$ and $b \in \Act_m$.
	
	Now, it suffices to construct $\kappa \in \Qsetp$ such that
        $\kappa \leq \min\{\Qmin,\tmin,\rmax\}/2$ and the expected mean-payoff
        obtained by applying $\sigma$ to $\calM_i$ and $\calM_{i+1}$ differs
        by at most $\varepsilon/n$. 
	
	Now, we discuss one of the changes, say in $m_i=m$, i.e., we omit
        the indexes for the sake of simplicity. For every $m' \in M$, let $\Ex^\sigma[\Rew(m,m')]$ and
        $\Ex^\sigma[\Time(m,m')]$ be the expected 
	cost and time accumulated along a run initiated in $m$ before visiting $m'$ 
	(note that every $m'$ is eventually visited by a run initiated in $m$ with probability one). 
	As $\calM$ is strongly connected, the mean-payoff value $\Ex[\MP^{\sigma}]$ achieved by $\sigma$ in $\calM$ can be  expressed as
	\begin{align}
	& \Ex[\MP^{\sigma}]  =
	\frac{\Ex^\sigma[\Rew(m,m)]}{\Ex^\sigma[\Time(m,m)]}.
	\end{align}
	To use Lemma~\ref{lem:frac_bouds}, we need to provide an upper bound on the numerator and upper and lower bounds on the denominator. 
	The expected number of transitions required to visit $m''$ from $m'$ can be bounded from above by $(1/\Qmin')^n$. Hence,
	\begin{itemize}
		\item $\Ex^\sigma[\Rew(m',m'')] \ \leq \ \rmax' \cdot (1/\Qmin')^n$, and
		\item $\tmin' \ \leq \ \Ex^\sigma[\Time(m',m'')] \ \leq \ \tmax' \cdot (1/\Qmin')^n$~.
	\end{itemize}
    Moreover, according to Lemma~\ref{lem:frac_bouds} we need to bound the changes in the numerator and denominator.
    For this purpose, observe that 
	\begin{align}
	 \Ex[\MP^{\sigma}]  &=
	\frac{\Ex^\sigma[\Rew(m,m)]}{\Ex^\sigma[\Time(m,m)]}
	 = \nonumber \\[1ex]
	& = \frac{c(\sigma(m)) + \sum_{m'\in M \setminus\{m\}} Q(\sigma(m))(m') \cdot \Ex^\sigma[\Rew(m',m)]}
	{t(\sigma(m)) + \sum_{m'\in M \setminus\{m\}} Q(\sigma(m))(m') \cdot \Ex^\sigma[\Time(m',m)]}. \label{eq-MP_change}
	\end{align}
	Note that for $m'\neq m$, $\Ex^\sigma[\Rew(m,m')]$ and 
	$\Ex^\sigma[\Time(m,m')]$ do not depend 
	on the actions (and their $Q$, $t$, and $c$ values) of $\Act_m$. 
	Hence, these values do not change when modifying $Q(b)$, $t(b)$, 
	and $c(b)$ only for $b \in \Act_m$, and we can treat them as constants.
	If $Q(\sigma(m))(m')$, $t(\sigma(m))$, and $c(\sigma(m))$ change at 
	most by $\kappa$, then the numerator and the denominator of the above fraction~\eqref{eq-MP_change} change at most by 
	$\kappa + n \cdot \kappa \cdot \rmax'/(\Qmin')^n$ and $\kappa + n \cdot \kappa \cdot \tmax'/(\Qmin')^n$, respectively. I.e., we use maximum to bound both of them.
	By Lemma~\ref{lem:frac_bouds}, the error of the fraction is bounded by $\varepsilon/n$ if 
	\[
	\kappa + n \cdot \kappa \cdot w'_{max}/(\Qmin')^n\ \leq \ 
	\frac{(\tmin')^2 \cdot \frac{\varepsilon}{n}}%
	{\rmax' \cdot (1/\Qmin')^n + \tmax' \cdot (1/\Qmin')^n \cdot (1+ \frac{\varepsilon}{n})},
	\]
	where $w'_{max} = \max\{\rmax',\tmax'\}$, that we can strengthen to
	\[
		\kappa + n \cdot \kappa \cdot w'_{max}/(\Qmin')^n\ \leq \ 
		\frac{(\tmin')^2 \cdot \frac{\varepsilon}{n}}%
		{w'_{max} \cdot (1/\Qmin')^n \cdot (2+ \frac{\varepsilon}{n})},
	\] 
	i.e.,
	\[
	\kappa \ \leq \ 
	\frac{(\tmin')^2 \cdot \frac{\varepsilon}{n}}%
	{w'_{max} \cdot (1/\Qmin')^n \cdot (2+ \frac{\varepsilon}{n})\cdot (1 + n \cdot w'_{max}/(\Qmin')^n)}.
	\] 
	Using the bounds of the original semi-MDP and the above mentioned restriction on $\kappa$, we obtain
    $$
	\kappa\ \leq \ \min \left\{ 
	\begin{array}{l}
	\frac{(\tmin/2)^2 \cdot \frac{\varepsilon}{n}}%
	{2 w_{\max} \cdot (2/\Qmin)^n \cdot (2+\frac{\varepsilon}{n}) \cdot (1+2n w_{\max} (2/\Qmin)^n)}, 
	\frac{\Qmin}{2}, \frac{\tmin}{2}, \frac{\rmax}{2} 
	\end{array} 
	\right\}
	$$
	where $w_{\max} = \max\{\rmax,\tmax\}$.
	 \qed
\end{proof}

Note that Assumption~\ref{ass4} is a trivial consequence of Assumtion~\ref{app-assam-a}. 
Hence, we can conclude with the following corollary.  

\begin{corollary}
	If Assumptions~\ref{app-assam-a} and \ref{app-assam-b} are fulfilled, then also Assumptions~\ref{ass1} and \ref{ass4} are fulfilled. 
\end{corollary}
  
\subsection{Dirac Distributions}
\label{sec:Dirac}

We start with showing that the Assumptions are fulfilled for Dirac distributions. The results will then be used in the analysis of other distributions as well.   
We assume a fixed $s \in \Sset$ such that $s \in S_a$ where \mbox{$F_a[d](\tau) = 1$} for all $\tau \geq d$ and $F_a[d](\tau) = 0$ for all  $0 \leq \tau < d$, where $d \in [\ell_a,u_a] \subset (0,\infty)$.\footnote{Note that we need to restrict $\ell_a$ and $u_a$ to work with correct parametric ACTMC.}

 Note that for Dirac distributions, we can easily obtain $\semiProb_s(d)$ by employing a Poisson distribution ranging over the number of exponentially-distributed delay transitions until time $d$. Then,
\[
   \semiProb_s(d)  \quad = \quad \sum_{i=0}^{\infty} e^{-\lambda d}\frac{(\lambda d)^i}{i!} \cdot \left( \vec{1}_s \cdot \overline{P}^i \right) \cdot P_a 
\]
where $i$ represents the number of exponential transitions fired before the alarm rings (i.e., time $d$), $e^{-\lambda d}\frac{(\lambda d)^i}{i!}$ is the corresponding probability according to the Poisson distribution,  $\vec{1}_s$ is a vector of zeroes except for $\vec{1}_s(s) = 1$, and $\overline{P} \colon S {\times} S\rightarrow [0,1]$ is a probability matrix that is as $P$ but where all states
in $\Soff$ are made absorbing, i.e., $\overline{P}(s_1,\cdot) =  P(s_1,\cdot)$ for all $s_1 \in S \setminus \Soff$, 
and $\overline{P}(s_1,s_1) = 1$ for all $s_1 \in \Soff$. 

Similarly, $\mcost_s$ is the expected total cost computable as follows:
\begin{align*}
	\mcost_s(d)  =  \sum_{i=0}^{\infty}
	e^{-\lambda d}\frac{(\lambda d)^i}{i!} 
	& \bigg( ~
	\sum_{j=0}^{i-1} \left(\vec{1}_s \cdot \overline{P}^j\right) \cdot
	\left(\frac{d\cdot\overline{\calR}}{i+1} + \overline{\calI}
	\right)\\ 
	&  \;+\;
	\left(\vec{1}_s \cdot \overline{P}^i\right) \cdot
	\left(\frac{d\cdot\overline{\calR}}{i+1} +
	\overline{\calI}_A\right) \bigg)  
\end{align*}
where $\overline{\calR}$ is a vector which is the same as $\calR$ but returns~$0$ for all states of $\Soff$, and $\overline{\calI},\overline{\calI}_A \colon S \to \Rsetpo$ are vectors that assign to each state the expected instant execution cost of the next delay and the next alarm transition, respectively.

Note that $\Theta_s$ is a special case of $\mcost_s$ where $\calR(s) = 1$ for all $s \in S$, and $\calI,\calI_A$ return zero for all arguments. 

Although the functions $\semiProb_s(\hat{\tau})$, $\mcost_s(\hat{\tau})$, and $\Theta_s(\hat{\tau})$ are defined as infinite sums, for every $\hat{\tau} \in \Qsetp$ and 
$\kappa \in \Qsetp$,
 a $k \in \Nset$ can be effectively computed such that for all $s \in \Sset$ and $d \leq \hat{\tau}$, the truncated versions of $\semiProb_s$, $\mcost_s$, and $\Theta_s$, obtained by taking the first $k$ summands of the corresponding defining series under-approximate the values of $\semiProb_s(d)$, $\mcost_s(d)$, and $\Theta_s(d)$ up to an absolute error of at most~$\kappa$. 
 We restrict $k$ to be always larger than $|S|$, thus we detect all non-zero probabilities and all distributions will have the same support.
 We use 
 $$\semiProb_s^{\hat{\tau},\kappa}, \mcost_s^{\hat{\tau},\kappa} \text{, and }\Theta_s^{\hat{\tau},\kappa}$$ 
 to denote these truncated versions that are $\kappa$-approximations for all values $d$ up to $\hat{\tau}$. Note that, due to the factor $e^{-\lambda d}$, the values of $\semiProb_s^{\hat{\tau},\kappa}$, $\mcost_s^{\hat{\tau},\kappa}$, and $\Theta_s^{\hat{\tau},\kappa}$ are still irrational even for rational $d$'s. Nevertheless, these values can be effectively approximated by rational numbers up to an arbitrarily small given error. Also note that the components of $\semiProb_s^{\hat{\tau},\kappa}$
may not sum up to~$1$. When we need to understand  $\semiProb_s^{\hat{\tau},\kappa}$ as a distribution, the remaining probability is split evenly among the states where $\semiProb_s^{\hat{\tau},\kappa}$ is positive.

 We show that Assumptions~\ref{app-assam-a} and \ref{app-assam-b}, as well as Assumptions~\ref{ass2}~and~\ref{ass3}\ of Section~\ref{sec-mdp}, are satisfied.
\smallskip
~\\\textbf{Assumption~\ref{app-assam-a}.}
Deriving the bounds $\semiProb_s^{\min}$, $\Theta_s^{\min}$, $\Theta_s^{\max}$, and $\mcost_s^{\max}$ is easy. We put
\begin{itemize}
	\item $\semiProb_s^{\min} = (P_{\min})^n \cdot \min \left\{ \frac{e^{-\lambda d} (\lambda d)^k}{k!} : 0 {\leq} k {\leq} n, d \in \{\ell_a,u_a\} \right\}$
	\item $\Theta_s^{\min} = \int_{0}^{\ell_a} x \cdot \lambda \cdot  e^{-\lambda x} \,\textrm{d}x + \ell_a \cdot e^{-\lambda \ell_a} = \frac{1-e^{-\lambda\cdot\ell_a}}{\lambda}$,
	\item $\Theta_s^{\max} = u_a$,
	\item $\mcost_s^{\max} = \calR_{\max} \cdot u_a + \calI_{\max} \cdot (\lambda\,u_a +1)$,
\end{itemize}
where $n=|S|$, $P_{\min} = \min\{P(s,s'),P_a(s,s') : s,s' \in S\}$,  $\calR_{\max} = \max\{\calR(s) : s \in S\}$ and $\calI_{\max} = \max\{\calI(s,s'),\calI_A(s,s') : s,s' \in S\}$. Note that although some of the defining expression involve the function $e^{x}$, we can still effectively under-approximate their values by positive rationals.
\smallskip
~\\\textbf{Assumption~\ref{app-assam-b}.} To define an appropriate $\delta_{(s,\kappa)} \in \Qsetp$, first note that for every $d \in [\ell_a,u_a]$ and every $\delta \in \Qsetp$ such that $d-\delta \in [\ell_a,u_a]$, we have that
\begin{itemize}
	\item $|\Theta_s(d) - \Theta_s(d{-}\delta)| \leq \delta$,
	\item $|\mcost_s(d) - \mcost_s(d{-}\delta)| \leq \delta\calR_{\max} + 2\lambda\delta\calI_{\max}$,
	\item $|\semiProb_s(d)(s') - \semiProb_s(d{-}\delta)(s')| \leq \lambda\delta$,
\end{itemize}
where $\calR_{\max}$ and $\calI_{\max}$ are as defined above. Note also that $\delta\calR_{\max}$ bounds the change of rate cost caused by changing $d$ by~$\delta$, $\lambda\delta$ bounds the change of the number of
delay transitions in time $d$ vs.{} time $d-\delta$, and $\lambda\delta$ can be
also used as a bound on the change of probabilities in the transient
probability distribution in CTMCs with rate $\lambda$ after time $d$
vs.{} $d-\delta$ (see, e.g., Theorem 2.1.1 of \cite{Norris:book}).
Hence, for a given $\kappa > 0$, we can set
\[
    \delta_{(s,\kappa)} \ \ =\ \ \min\left \{\kappa, \frac{\kappa}{\calR_{\max} + 2\lambda\calI_{\max}} \right\}.
\]

~\\\textbf{Assumption~\ref{ass2}).}
Let $\hat{\tau}=u_a$. 
The functions  $\semiProb_s^\kappa$, $\Theta_s^\kappa$, and $\mcost_s^\kappa$ are defined as rational $\kappa/2$-approximations of $\semiProb_s^{\hat{\tau},\kappa/2}$, $\mcost_s^{\hat{\tau},\kappa/2}$, and $\Theta_s^{\hat{\tau},\kappa/2}$ that are later approximated (for which we use the remaining $\kappa/2$ error) by rational functions using truncated Taylor series for $e^{(\cdot)}$.

~\\\textbf{Assumptions~\ref{ass3}).}
 For $\hat{\tau}=u_a$, we put 
\[ 
  \pmb{\semiProb}_s^\kappa = \semiProb_s^{\hat{\tau},\kappa}, \quad \pmb{\mcost}_s^\kappa = \mcost_s^{\hat{\tau},\kappa}, \quad \pmb{\Theta}_s^\kappa = \Theta_s^{\hat{\tau},\kappa}, and
\] 
\[
\pmb{F}_s^\kappa[g,\mathbf{h}](d) \ =\  \pmb{\mcost}^\kappa_s(d) - g \cdot \pmb{\Theta}^\kappa_s(d) + \pmb{\semiProb}^\kappa_s(d) \cdot  \mathbf{h},
\]
where $g \in \Qset$ and $\mathbf{h}\colon \Sset \cup \Soff \rightarrow \Qset$ are constant
(cf. Equation~(\ref{eqn:Fapprox})). Observe that
$\pmb{\mcost}^\kappa_s(d)$, $\pmb{\Theta}^\kappa_s(d)$, and $\pmb{\semiProb}^\kappa_s(d) \cdot  \mathbf{h}$ are of the form $e^{-\lambda d} \cdot \mathit{Poly}(d)$, where $\mathit{Poly}(d)$ is a polynomial. Hence, the first derivative of $\pmb{F}_s^\kappa[g,\mathbf{h}](d)$  also has the special form $e^{-\lambda d} \cdot V_{s,g,\mathbf{h}}$ where $V_{s,g,\mathbf{h}}$ is a polynomial of one free variable. Since $e^{-\lambda d}$ is positive for $d \in \Rsetpo$, the derivative $e^{-\lambda d} \cdot
V_{s,g,\mathbf{h}}$ is zero iff $V_{s,g,\mathbf{h}}$ is zero. Hence, the roots of the derivative in the interval $[\ell_a,u_a]$ can be approximated by approximating the roots of $V_{s,g,\mathbf{h}}$, 
which is computationally easy (e.g., using tools such as $\maple$ \cite{maple}). 

\subsection{The Symbolic Functions for Action Effects with non-Dirac distribution}  
\label{sec:functions}

The main idea repeatedly applied in the following subsections is that, e.g.,
the cost function $\Theta_s(d)$ for a non-Dirac continuous distribution $F_a[d]$ with a known density function $f_a[d]\colon \Rset \to \Rsetpo$ can be computed by integration of the cost function $\text{Dirac-}\Theta_s(\cdot)$ for some Dirac distribution applying the density function as a measure. That is,
\begin{equation}
\Theta_s(d)= \int_{min_a}^{max_a}\! \text{Dirac-}\Theta_s(\tau) \cdot f_a[d](\tau) \;\mathrm{d}{\tau},\label{eq:integration-of-dirac}
\end{equation}
where $\tau$ stands for the randomly chosen time according to the non-Dirac distribution and $min_a$ and $max_a$ bound the support of $f_a[d]$.
Note that, for every (even non-continuous) distributions we can obtain $\Theta_s(d)$ by employing Lebesgue integral, i.e., let $\mu_a$ be the measure of $F_a[d]$, then
$$ \Theta_s(d)=\int_{[min_a,max_a]}\! \text{Dirac-}\Theta_s(\tau) \;\mu_a(\mathrm{d} \tau).$$
In following subsections, we denote the following values computed for the Dirac distribution (see last section)
$\semiProb_s$, $\Theta_s$, 
$\mcost_s$, $\semiProb_s^{\hat{\tau},\kappa}$, $\Theta_s^{\hat{\tau},\kappa}$, $\mcost_s^{\hat{\tau},\kappa}$,
$\semiProb_s^{\min}$,
$\Theta_s^{\max}$, and $\mcost_s^{\max}$ by
$\text{Dirac-}\semiProb_s$, $\text{Dirac-}\Theta_s$, $\text{Dirac-}\mcost_s$, $\text{Dirac-}\semiProb_s^{\hat{\tau},\kappa}$, $\text{Dirac-}\Theta_s^{\hat{\tau},\kappa}$,  $\text{Dirac-}\mcost_s^{\hat{\tau},\kappa}$,
$\text{Dirac-}\semiProb_s^{\min}$, 
$\text{Dirac-}\Theta_s^{\max}$ and $\text{Dirac-}\mcost_s^{\max}$, respectively. 
\subsection{Uniform Distributions}
\label{sec:uniform}

In a (continuous) uniform distribution there are two parameters, the lower bound on possible random values, say $\alpha$, and the upper bound, say $\beta$. 
Let us note that using a smart construction, we can easily synthesize both of the parameters by our algorithm. 
Note that the first parameter $\alpha$ can be modeled as an alarm with Dirac distribution (parametrized by $\alpha$ that we can synthesize) that subsequently enables a uniformly distributed alarm with the first parameter fixed to 0 and the second parameter, say $\beta'$, that is also subject of synthesis. 
Then, the required parameters of the original uniform distribution are $\alpha$ and $\beta=\alpha+\beta'$.
Note that the newly created uniformly distributed alarm may not be localized, what can rule out the applicability of our synthesis algorithm.
However we are working on extension of our algorithm to allow for ACTMC with non-localized alarms. 

Therefore, in what follows we consider alarms that are \emph{uniformly} distributed on a time interval starting in 0 with a parametrized length. That is, we assume a fixed $s \in \Sset$ such that $s \in S_a$ where 
\mbox{$F_a[d](\tau) = \frac{\tau}{d}$} for all $0 \leq \tau \leq d$,
$F_a[d](\tau) = 0$ for all  $\tau < 0$, and $F_a[d](\tau) = 1$ for all  $\tau > d$, where $d \in [\ell_a,u_a] \subset (0,\infty)$.\footnote{Note that we need to restrict $\ell_a$ and $u_a$ to work with correct parametric ACTMC.}

~\\\textbf{Assumption~\ref{app-assam-a}.}
Note that, for each $d \in [\ell_a,u_a]$ the uniform distribution has its support at most on $[0,u_a]$.
Moreover, 
$\text{Dirac-}\Theta_s$, $\text{Dirac-}\mcost_s$ 
are bounded from above by 
$\text{Dirac-}\Theta_s^{\max}$, $\text{Dirac-}\mcost_s^{\max}$
on $[0,u_a]$.
Thus also the integrals for defining $\Theta_s$ and $\mcost_s$ are bounded and we have:
\begin{itemize}
	\item $\Theta_s^{\max} \eqdef  \text{Dirac-}\Theta_s^{\max} = u_a$ and
	\item $\mcost_s^{\max} \eqdef \text{Dirac-}\mcost_s^{\max}=\calR_{\max} \cdot u_a + \calI_{\max} \cdot (\lambda\,u_a +1)$.
\end{itemize}
Note that, for each $d \in [\ell_a,u_a]$, the uniform distribution has at least $1/2$ of the probability on $[\ell_a/2,u_a]$.
Thus, we can use half of $\text{Dirac-}\semiProb_s^{\min}$ and $\text{Dirac-}\Theta_s^{\min}$ for Dirac distribution on $[\ell_a/2,u_a]$:
\begin{itemize}
	\item $\semiProb_s^{\min} = (P_{\min})^n \cdot \min \left\{ \frac{e^{-\lambda d} (\lambda d)^k}{2 \cdot k!} : 0 {\leq} k {\leq} n, d \in \{\ell_a/2,u_a\} \right\}$ and
	\item $\Theta_s^{\min} = \frac{1-e^{-\lambda\cdot\ell_a/2}}{2\lambda}$.
\end{itemize}

\smallskip 
~\\\textbf{Assumption~\ref{app-assam-b}.} First, consider the bound on $|\mcost_s(d) - \mcost_s(d{-}\delta)|$. Note that $1/d$ is the density of $F_a[d]$ on interval $[0,d]$. We have that
\begin{align*}
& |\text{Dirac-}\mcost_s(d) - \text{Dirac-}\mcost_s(d{-}\delta)| \ =\\[.5ex]
=\ &  \left| \int_{0}^{d} \frac{\text{Dirac-}\mcost_s(\tau)}{d}\, \textrm{d}\tau
- \int_{0}^{d-\delta} \frac{\text{Dirac-}\mcost_s(\tau)}{d-\delta} \, \textrm{d}\tau \,\right|\\[.5ex]
=\ &  \int_{0}^{d} \frac{\text{Dirac-}\mcost_s(\tau)}{d}\, \textrm{d}\tau
- \int_{0}^{d-\delta} \frac{\text{Dirac-}\mcost_s(\tau)}{d-\delta} \, \textrm{d}\tau \tag*{as Dirac-$\mcost_s$ is nondecreasing }\\[.5ex]
=\ & \int_{d-\delta}^{d} \frac{\text{Dirac-}\mcost_s(\tau)}{d}\, \textrm{d}\tau 
+ \int_{0}^{d-\delta} \frac{\text{Dirac-}\mcost_s(\tau)}{d} - \frac{\text{Dirac-}\mcost_s(\tau)}{d-\delta}\, \textrm{d}\tau
\,\\[.5ex]
\leq\ & \int_{d-\delta}^{d} \frac{\text{Dirac-}\mcost_s(\tau)}{d}\, \textrm{d}\tau \, \tag*{as $\frac{1}{d} - \frac{1}{d-\delta}$ is negative }\\[.5ex]
\leq\ & \frac{1}{d}\int_{d-\delta}^{d} \calR_{\max} \cdot u_a + \calI_{\max} \cdot (\lambda\,u_a +1)\, \textrm{d}\tau \tag*{as Dirac-$\mcost_s(\tau) \leq $ Dirac-$\mcost_s^{\max}$}\\[.5ex]
\leq\ & \delta\cdot(\calR_{\max} \cdot u_a + \calI_{\max} \cdot (\lambda\,u_a +1))/\ell_a \tag*{as $d \geq \ell_a$}
\end{align*}

Similarly we can get a bound on $|\Theta_s(d) - \Theta_s(d{-}\delta)| \leq \delta\cdot u_a/\ell_a$.
Now we consider a bound on $|\semiProb_s(d)(s') - \semiProb_s(d{-}\delta)(s')|$:
\begin{align*}
& |\semiProb_s(d)(s') - \semiProb_s(d{-}\delta)(s')| \ =\\[.5ex]
=\ &  \left| \int_{0}^{d} \frac{\text{Dirac-}\semiProb_s(\tau)(s')}{d}\, \textrm{d}\tau
- \int_{0}^{d-\delta} \frac{\text{Dirac-}\semiProb_s(\tau)(s')}{d-\delta} \, \textrm{d}\tau \,\right|\\[.5ex]
=\ &  \left| \int_{d-\delta}^{d} \frac{\text{Dirac-}\semiProb_s(\tau)(s')}{d}\, \textrm{d}\tau
+ \left(\frac{1}{d} -\frac{1}{d - \delta}\right) \cdot \int_{0}^{d-\delta} \text{Dirac-}\semiProb_s(\tau)(s')\, \textrm{d}\tau \,\right|\\[.5ex]
=\ &  \left| \int_{d-\delta}^{d} \frac{\text{Dirac-}\semiProb_s(\tau)(s')}{d}\, \textrm{d}\tau
- \frac{\delta}{d^2 - d\delta} \cdot \int_{0}^{d-\delta} \text{Dirac-}\semiProb_s(\tau)(s')\, \textrm{d}\tau \,\right|\\[.5ex]
=\ &  |G - H| \text{, where} 
\end{align*}
$G \eqdef \int_{d-\delta}^{d} \frac{\text{Dirac-}\semiProb_s(\tau)(s')}{d}\, \textrm{d}\tau$ and $H \eqdef \frac{\delta}{d^2 - d\delta} \cdot \int_{0}^{d-\delta} \text{Dirac-}\semiProb_s(\tau)(s')\, \textrm{d}\tau$.\\
 If $G < H$ then 
\begin{align*}
|\semiProb_s(d)(s') - \semiProb_s(d{-}\delta)(s')| \ &\leq\  H = \frac{\delta}{d^2 - d\delta} \cdot \int_{0}^{d-\delta} \text{Dirac-}\semiProb_s(\tau)(s')\, \textrm{d}\tau\\[.5ex]
&\leq\  \frac{\delta}{d^2 - d\delta} \cdot  \int_{0}^{d-\delta} 1\, \textrm{d}\tau\\[.5ex]
&  = \frac{\delta}{d^2-d\delta} \cdot (d-\delta)  = \frac{\delta}{d} \leq \frac{\delta}{\ell_a}.
\end{align*}
If $G \geq H$ then 
$$|\semiProb_s(d)(s') - \semiProb_s(d{-}\delta)(s')| 
\ \leq\ \frac{1}{d} \int_{d-\delta}^{d} \text{Dirac-}\semiProb_s(\tau)(s')\, \textrm{d}\tau 
\ \leq\ \frac{1}{d} \int_{d-\delta}^{d} 1\, \textrm{d}\tau = \frac{\delta}{d} \leq \frac{\delta}{\ell_a}.$$
Hence, we need
\[
 \max\left\{\delta\cdot\frac{\calR_{\max} \cdot u_a + \calI_{\max} \cdot (\lambda\,u_a +1)}{\ell_a}, \delta\cdot \frac{u_a}{\ell_a}, \frac{\delta}{\ell_a}\right \} \leq \kappa .
\]
Thus, for a given $\kappa > 0$, we can set
\[
\delta_{(s,\kappa)} \ \ =\ \ \min\left \{\frac{\kappa \cdot \ell_a}{\calR_{\max} \cdot u_a + \calI_{\max} \cdot (\lambda\,u_a +1)}, \frac{\kappa \cdot \ell_a}{u_a}, \kappa \cdot \ell_a \right\}.
\]

\smallskip

~\\\textbf{Assumptions~\ref{ass2}) and \ref{ass3}).} 
We show that $\mcost_s$ can be \mbox{$\kappa$-approximated} by a function  $\pmb{\mcost}_s^\kappa(d)$ of the form $V(d)/d$, where $V$ is a computable polynomial of rationals coefficients. 
By using the same technique, we can construct also  $\pmb{\Theta}_s^\kappa$ and  $\pmb{\semiProb}_s^\kappa$, which are of the same form. 
Hence, the function $\pmb{F}_s^\kappa[g,\mathbf{h}](d) = \pmb{\mcost}^\kappa_s(d) - g \cdot \pmb{\Theta}^\kappa_s(d) + \pmb{\semiProb}^\kappa_s(d) \cdot  \mathbf{h}$ 
has also the same form, since $g$ and $\mathbf{h}$ are constant. 
After differentiation of $\pmb{F}_s^\kappa[g,\mathbf{h}](d)$ we get a function of form $V'(d)/d^2$, where we isolate roots of polynomial $V'(d)$ to obtain the candidate actions.
Note that the denominator $d^2$ may disable root $d=0$ that is not within the eligible parameter values.
The correctness of the Algorithm~\ref{alg:pol-iter} is not harmed by including candidate actions computed from the ``false positive'' root, since some other action of the candidate set must outperform all of them. 
Thus, the function $\pmb{F}_s^\kappa[g,\mathbf{h}](d)$ fulfills Assumption~\ref{ass2} and~\ref{ass3}.

Now we show that $\mcost_s$ can be \mbox{$\kappa$-approximated} by a function  $\pmb{\mcost}_s^\kappa(d)$ of the form $V(d)/d$, where $V$ is a computable polynomial of rationals coefficients.
From definition $\mcost_s(d)$ equals to $\int_{0}^{d} 1/d \cdot  \text{Dirac-}\mcost_s(\tau) \,\textrm{d}\tau $. 
Instead of $\text{Dirac-}\mcost_s(\tau)$ we can use its $\kappa/2$ approximation $\text{Dirac-}\mcost_s^{u_a,\kappa/2}(\tau) = P(\tau) \cdot e^{-\lambda \tau}$, where $P(\tau)$ is an univariate polynomial with rational coefficients.
The main difficulty is an integration of $P(\tau) \cdot e^{-\lambda \tau}$.
Fortunately, we can use the Taylor series representation for $e^{-\lambda \tau}$, i.e.,
\[
e^{-\lambda \tau} = \sum_{n=0}^{\infty} \frac{1}{n!} \cdot (-\lambda\tau)^n.
\]
Then, we can truncate the infinite sum, such that we cause at most $\kappa/2$ error in $P(\tau) \cdot e^{-\lambda \tau}$ for each $\tau \in [0,u_a]$.
Thus we obtain polynomial $P(\tau) \cdot S(\tau)$ with rational coefficients, that $\kappa$-approximates $\text{Dirac-}\mcost_s(\tau)$.
Hence we can set
\[
\pmb{\mcost}_s^\kappa(d)\ \ 
=\ \frac{1}{d} \cdot \int_{0}^{d} P(\tau) \cdot S(\tau) \, \,\textrm{d}\tau
=\ \frac{1}{d} \cdot V(d)
\]
and we are done.

\subsection{Weibull and Other Continuous Distributions}
\label{sec:Weibull}

As we already noted in Section~\ref{sec-prelims}, the abstract assumptions enabling the applicability of our symbolic algorithm are satisfied also by other distributions. 
Since the arguments are increasingly more involved (and not used in our experiments reported in Section~\ref{sec:experiments}), we just sketch the arguments that show how to handle distributions with polynomial approximations and illustrate them on the Weibull distribution.
Note that exponential distribution is a special case of the Weibull distribution, where the  
fixed shape constant $k=1$.

~\\\textbf{Assumption~\ref{app-assam-a}.}
Let $\pmb{f}_a[d](\tau)$ be the derivative of $\pmb{F}_a[d](\tau)$, i.e., a probability density function. Let $min_a'$ and $max_a'$ be bounds such that for all $d \in [\ell_a,u_a]$ we have that 

\begin{align*}
\int_{0}^{\infty} \pmb{f}_a[d](\tau) \cdot \langle ~.~ \rangle_{max} \,\textrm{d}\tau - \int_{min_a'}^{max_a'} \pmb{f}_a[d](\tau) \cdot \langle ~.~ \rangle_{max} \,\textrm{d}\tau  \ &\leq \ \frac{1}{2}
\end{align*}

where $\langle ~.~ \rangle_{max}$ is 
\begin{itemize}
	\item $1$ for $\semiProb_s$,
	\item $\text{Dirac-}\Theta_s^{\max} = max_a'$ for $\Theta_s$,
	\item $\text{Dirac-}\mcost_s^{\max} = \calR_{\max} \cdot max_a' + \calI_{\max} \cdot (\lambda\,max_a' {+}1)$ for $\mcost_s$,
\end{itemize}
and $\calR_{\max}$ and $\calI_{\max}$ are the same as in Section~\ref{sec:Dirac}. Note that such bounds exist because the density function has to decrease in its extrema faster than at most linearly increasing (in $max_a'$) functions $\text{Dirac-}\Theta_s^{\max}$ and $\text{Dirac-}\mcost_s^{\max}$.

We can safely set 
\begin{itemize}
	\item $\Theta_s^{\max} = \text{Dirac-}\Theta_s^{\max} + 1/2= max_a' + 1/2$~ and
	\item $\mcost_s^{\max} = \text{Dirac-}\mcost_s^{\max} + 1/2 =\calR_{\max} \cdot max_a' + \calI_{\max} \cdot (\lambda\,u_a +1) + 1/2$.
\end{itemize}
Let $c \eqdef \min_{d \in [\ell_a,u_a]} \int_{min_a'}^{max_a'} \pmb{f}_a[d](\tau) \,\textrm{d}\tau$. One can easily bound $c$ from below by $1/2$ or more exactly by $(1-1/(2\cdot \max \{1, \text{Dirac-}\Theta_s^{\max}, \text{Dirac-}\mcost_s^{\max} \}))$.
Using $c$ we can set 
\begin{align*}
\semiProb_s^{\min} 
&= c \cdot \text{Dirac-}\semiProb_s^{\min} \\
&= c \cdot (P_{\min})^n \cdot \min \left\{ \frac{e^{-\lambda d} (\lambda d)^k}{k!} : 0 {\leq} k {\leq} n, d \in \{min_a',max_a'\} \right\}
\end{align*}
and
$$\Theta_s^{\min} = c \cdot \text{Dirac-}\Theta_s^{\min} = c \cdot \frac{1-e^{-\lambda\cdot max_a'}}{2\lambda}.$$

\smallskip \smallskip 

~\\\textbf{Assumptions~\ref{ass2}) and \ref{ass3}).} Employing Assumption~\ref{app-assam-a}, we can compute sufficiently small $\kappa$.
Now, %
we define $max_a$ to be bound such that for all $d \in [\ell_a,u_a]$ we have that 

\begin{align*}
1 - \int_{0}^{max_a} \pmb{f}_a[d](\tau) \,\textrm{d}\tau %
\ &\leq \ \frac{\kappa}{2\cdot \langle ~.~ \rangle_{max}}
\end{align*}
where $\langle ~.~ \rangle_{max}$ is same as above, but we use $max_a$ instead of $max_a'$.

Now, let $U(\tau)$ be a polynomial approximation of $\pmb{f}_a[d](\tau) \cdot \text{Dirac-}\mcost_s(\tau)$ whose error is bounded by ${\kappa}/{(2 \cdot max_a)}$ on interval $[0, max_a]$. 
Then $\pmb{\mcost}_s^\kappa = \int_{0}^{max_a} U(\tau)\,\textrm{d}\tau$ is a  $\kappa$-approximation of $\mcost_s$.

Note that $U(\tau)$ is a polynomial with variable $\tau$, but might not be polynomial in the parameter $d$. 
Thus, there is a question whether $\int_{0}^{max_a} U(\tau)\,\textrm{d}\tau$ is differentiable for every $d \in [\ell_a,u_a]$. As a concrete example, let us consider a Weibull distribution with a fixed shape parameter $k \in \Nset$ and the scale parameter $1/d$ (i.e., we aim at synthesizing an $\varepsilon$-optimal scale). 
The Weibull density function is 
\[
  k  d \cdot (\tau d )^{k-1} \cdot e^{-(\tau  d )^k}.
\] 
As we know from Subsection~\ref{sec:Dirac}, $\text{Dirac-}\mcost_s^{max_a,\kappa}(\tau)=P(\tau) \cdot e^{-\lambda\tau}$, where $P(\tau)$ is a polynomial function. 
Note that both $e^{-(\tau d )^k}$ and $e^{-\lambda\tau}$ can be approximated to arbitrary precision for $d \in [\ell_a,u_a]$ and $\tau \in [0,max_a]$ using a finite prefix of the Taylor series 
\[
 \sum_{n=0}^{\infty} \frac{1}{n!} \cdot (-(\tau d)^k  -\lambda\tau)^n.
\]
Then, for some sufficiently high bound $i\in\Nset$,
$$U(\tau)\ \ =\ \ 
kd \cdot (\tau d)^{k-1} \cdot  P(\tau) \cdot \sum_{n=0}^{i} \frac{1}{n!} \cdot (-(\tau d)^k  -\lambda\tau)^n
$$
is a polynomial function of $\tau$ and hence, %
$\pmb{\mcost}_s^\kappa = \int_{0}^{max_a} U(\tau)\,\textrm{d}\tau$ is a polynomial function of~$d$. 
Similarly we can provide polynomial $\kappa$-approximations $\pmb{\semiProb}_s^\kappa$ and $\pmb{\Theta}_s^\kappa$, thus Assumption~\ref{ass3} holds.
Moreover, $\pmb{\semiProb}_s^\kappa$, $\pmb{\Theta}_s^\kappa$, and $\pmb{\mcost}_s^\kappa$ are computable for each $d \in [\ell_a,u_a]$, thus Assumption~\ref{ass2} is also fulfilled.

\smallskip 
~\\\textbf{Assumption~\ref{app-assam-b}.} 
Similarly to the beginning of Section~\ref{sec:dist} we first choose sufficiently small $\kappa$ such that the optimal strategy for 
each $\calM \in \left[\calM_\calN\right]_{4\kappa}$ induces an $\varepsilon$-optimal strategy for $\calN$.
Using $\kappa$, we $\kappa$-approximate $\semiProb_s$, $\Theta_s$, $\mcost_s$ functions by polynomials $\pmb{\semiProb}_s^\kappa$, $\pmb{\Theta}_s^\kappa$, $\pmb{\mcost}_s^\kappa$. 
Then we use $\pmb{\semiProb}_s^\kappa$, $\pmb{\Theta}_s^\kappa$, and $\pmb{\mcost}_s^\kappa$ to obtain bounds on their first derivative using their second derivative and root isolation on interval $[\ell_a,u_a]$.
We use these bounds to obtain sufficiently small $\delta_{(s,\kappa)}$ (and thus $\delta\colon \Sset \rightarrow \Qsetp$) to cause at most $\kappa$ error.
The optimal strategy for any 
$\calM \in \left[\calM_\calN\langle\delta\rangle\right]_{\kappa}$ 
is also optimal strategy for some 
$\calM \in \left[\calM_\calN\right]_{3\kappa}$
and thus induces an $\varepsilon$-optimal strategy for $\calN$.

\end{document}